\documentclass[journal,11pt,draftclsnofoot,onecolumn]{IEEEtran}
\usepackage{graphicx}
\usepackage{epsf}
\usepackage{amsmath,amssymb,amsthm,bm}
\usepackage{verbatim}
\usepackage{color}
\usepackage{color,soul}
\usepackage[round]{natbib}
\usepackage[dvipsnames]{xcolor}

\usepackage{booktabs, multirow, arydshln}

\usepackage{comment}

\usepackage{epstopdf} 

\makeatletter
\@addtoreset{equation}{section}
\makeatother
\renewcommand{\theequation}{\thesection.\arabic{equation}}
\newtheorem{theorem}{\bf Theorem}[section]

\newtheorem{definition}[theorem]{Definition}
\newtheorem{example}[theorem]{Example}

\newtheorem{lemma}[theorem]{Lemma}

\newtheorem{proposition}[theorem]{Proposition}
\newtheorem{remark}[theorem]{Remark}

\graphicspath{./plot/}

\newcommand{\resethlcolor}{\sethlcolor{yellow}}
\resethlcolor

\newcommand{\resethlcolortwo}{\sethlcolor{Lavender}}
\resethlcolortwo

\begin{document}

\begin{center}
{\bf  R\'{e}nyi Divergence in General Hidden Markov Models}
\end{center} 
\vspace{0.3cm}
\begin{center}
 Cheng-Der Fuh \\
Fanhai International School of Finance, Fudan University \\
\vspace{0.1cm}
Su-Chi Fuh\\
Department of Computer Science and Information Engineering \\
National Taipei University of Technology \\
\vspace{0.1cm}
Yuan-Chen Liu \\
Department of Computer Science, National Taipei University of Education \\ 
\vspace{0.1cm}
Chuan-Ju Wang \\
Research Center for Information Technology Innovation, Academia Sinica  
\end{center}
\vspace{.5cm}

\centerline{\bf Abstract}
\vspace{0.2cm}
In this paper, we examine the existence of the R\'{e}nyi divergence between two
time invariant general hidden Markov models with arbitrary positive initial
distributions. By making use of a Markov chain representation of the
probability distribution for the general hidden Markov model and 
eigenvalue  for
the associated Markovian operator, we obtain, under some regularity conditions,
convergence of the R\'{e}nyi divergence. By using this device, we
also characterize the R\'{e}nyi divergence, and obtain the Kullback--Leibler
divergence as $\alpha \to 1$ of the R\'{e}nyi divergence. 
Several examples, including the classical finite state hidden Markov models,
Markov switching models, and recurrent neural networks, are given for
illustration. Moreover, we develop a non-Monte Carlo method that computes the
R\'{e}nyi divergence of two-state Markov switching models via the underlying
invariant probability measure, which is characterized by the Fredholm integral
equation. \\

\vspace{0.5cm}

~\\

\def\theequation{1.\arabic{equation}}
\setcounter{equation}{0}

\section{Introduction}\label{sec1}

Motivated by study of the information divergence  in hidden Markov models
(HMM), Markov switching models, and recurrent neural networks, 
we here investigate R\'{e}nyi divergences in general hidden Markov models.
A general hidden Markov model is, loosely speaking, a sequence ${\bf Y}=\{Y_n,
n \geq 0\}$ of random variables obtained in the following way. First, a
realization of a finite state Markov chain ${\bf X}=\{X_n, n \geq 0 \}$ is created.
This chain is sometimes called the regime and is not observed. Then,
conditioned on ${\bf X}$, the ${\bf Y}$-variables are generated. Usually, the
dependency of $Y_n$ on ${\bf X}$ is more or less local, as when $Y_n =
g(X_n,\eta_n)$ or $Y_n = g(X_n, Y_{n-1}, \eta_n)$ for some function $g$ and
random sequence $\{\eta_n \}$, independent of ${\bf X}$. $Y_n$ itself is
generally not Markov and may in fact have a complicated dependency structure.

HMMs have been studied extensively, and have 
extraordinary applications in fields as varied as speech
recognition, cf.\ \cite{Rabiner_Juang_1993}, \cite{Rabiner_1989}; handwritten
recognition, cf.\ \cite{Hu_Brown_1996, Kunda_1989}; human activity recognition,
cf.\ \cite{HAR}; target detection and tracking, cf.\ \cite{Willett,Fusion2010,
WSEAS2014};  computational molecular biology and bioinformatics, including DNA
and protein modeling, cf.\ \cite{Churchill_1989}; modeling, rapid detection,
and tracking of malicious activity of terrorist groups, cf.\ \cite{RGT2013},
\cite{Raghavanetal-IEEE2014}; and others, with just a small sample of
references given for each application. A comprehensive survey of HMM research
and applications can be found in \cite{Ephraim_Merhav_2002},
\cite{Cappe_Moulines_Ryden_2005}, and \cite{Zucchini_MacDonald_2009}, including
an extensive bibliography. 

A natural extension of the celebrated HMM is the following Markov swiching models.
We start with a simple real-valued first-order autoregression around one of two
constants $\mu_1$ or $\mu_2$:
\begin{eqnarray}\label{msm1}
Y_n = \mu_{X_n} + \psi Y_{n-1}  + \varepsilon_n,
\end{eqnarray}
where $\varepsilon_n \sim N(0,\sigma^2)$, $|\psi| < 1$, and $\{X_n, n \geq 0\}$
is a $2$-state Markov chain. 
When $\psi=0$, (\ref{msm1}) reduces to the classical Gaussian HMM.

Another interesting example is the recurrent neural network (RNN) in machine
learning. Note that the RNN can take as input a variable-length sequence $y=(y_1,
\cdots,y_n)$ by recursively processing each symbol while maintaining its
internal hidden state $h$. At each time step $n$, the RNN reads the symbol $Y_n
\in {\bf R}^q$ and updates its hidden state $h_n \in {\bf R}^p$ by
\begin{eqnarray}\label{RNN11}
h_n = f_{\theta} (Y_n,h_{n-1}),
\end{eqnarray}
where $f_\theta$ is a deterministic non-linear transition function, and
$\theta$ is the parameter of $f_\theta$. 
The transition function $f_\theta$ can be implemented with gated activation
functions such as long short-term memory (LSTM), 
cf.\ \cite{Hochreiter_Schmidhuber_1997}, or gated recurrent unit (GRU), 
cf.\ \cite{cho_et_al_2014}.
Although the hidden unit $h_n$ is a general state random variable, a
transformation in Section~\ref{RNN} shows that the RNN in (\ref{RNN11}) can be
formulated as a general HMM.

The  R\'{e}nyi divergence rate, cf.\ \cite{Renyi_1960}, and the Kullback--Leibler
divergence in particular,  have played a significant role in certain
hypothesis-testing questions, cf.\ \cite{Koopmans_1960}, \cite{Nemetz_1974}.
Furthermore, the R\'{e}nyi entropy and the R\'{e}nyi entropy rate have revealed
several operational characterizations in the problem of fixed-length source
coding, cf.\ \cite{Csiszar_1995}, \cite{Chen_Alajaji_2001}; unsupervised
learning, cf.\  \cite{Jenssen_Hild_Erdogmus_Principe_Eltoft_2003};
variable-length source coding, cf.\ \cite{Blumer_McEliece_1988},
\cite{Campbell_1965}, \cite{Jelinek_1968}, and
\cite{Rached_Alajaji_Campbell_2001}; error exponent calculations, cf.\
\cite{Erez_Zamir_2001}; policy optimization in reinforcement learning, cf.\
\cite{metelli2018policy}, \cite{papini2019optimistic}, and other areas such as
\cite{Arikan_1996}, \cite{Bassat_Raviv_1978}, and
\cite{Pronzato_Wynn_Zhigljavsky_1997}.

It is known that the root of the Kullback--Leibler divergence and the R\'{e}nyi
divergence is the celebrated Shannon entropy. 
The question of computing the Shannon entropy (or, simply, entropy) of a HMM
was studied in an early paper by \cite{Blackwell_1957}, in which the analysis
suggests the intrinsic complexity of expressing the HMM entropy as a function
of the process parameters. The author also presents an expression of the
entropy in terms of a measure~$Q$, which solves an integral equation dependent
on the parameters of the process. In general, the measure is hard to extract
from the equation in any explicit way. 
\cite{Fuh_Mei_2015} provide a numerical method to approximate the invariant
measure and the Kullback--Leibler divergence for a two-state HMM. The problem of
determining the residual noise of the best filter for a HMM was studied in
\cite{Khasminskii_Zeitouni_1996}, 
\cite{Ordentlich_Weissman_2004},
\cite{Ordentlich_Weissman_2006}, and \cite{Jacquet_Seroussi_Szpankowski_2008}
investigate the asymptotic estimates of the HMM entropy rate.
Furthermore, \cite{Zuk_Kanter_Domany_2005} present formulas for higher-order
coefficients of the Taylor expansion in the symmetric case.
\cite{Han_Marcus_2006B} and \cite{Han_Marcus_2006A} characterize the
analyticity of the HMM entropy rate, and obtain a broad generalization of the
results of \cite{Zuk_Kanter_Domany_2005}. 

For an explicit computation of  the R\'{e}nyi entropy of HMMs over finite
alphabets in both finite-length and asymptotic regimes, \cite{Wu_Xu_Han_2017}
discuss some convergence properties with no explicit formulas.
For Shannon entropy the problem has been found hard and solvable only for
specific cases, being related to an intractable task in random matrix 
products---finding top Lapunov exponents, cf.\ \cite{Jacquet_Seroussi_Szpankowski_2008}. 
In the case of finite state Markov chains, \cite{Rached_Alajaji_Campbell_2001}
apply substochastic matrices in the asymptotic regime powers, which can be
approximated by spectral analysis, to yield formulas on entropy rates. 

Although there are some interesting papers on computing the R\'{e}nyi divergence
in a special HMM, systematic study for a general HMM is still lacking. To fill
this gap, we investigate the R\'{e}nyi divergence for a general HMM in this
paper. We make three contributions. First, we note that a major
difficulty for analyzing the R\'{e}nyi divergence in  general HMMs is that the
joint probability  can be expressed only in summation form; see equations
(\ref{joint}) and (\ref{density}) in Section~\ref{sec2} for instance. The
constribution in this paper is that we provide a device which represents the
joint probability as the $L_1$-norm of a product of random matrices and treat
it as a Markov chain in an enlarged state space. This representation enables us
to apply results of the strong law of large numbers, and spectral theory for
the associated Markovian operator  of Markov random walks, to yield an explicit
characterization of the R\'{e}nyi divergence, and hence the Kullback--Leibler
divergence. Second, our formulation of the HMM in a general sense covers
several interesting examples, including finite state HMMs, Markov
switching models, and RNNs. Third, we develop a non-Monte Carlo method that
computes the R\'{e}nyi divergence of a two-state Markov switching model via the
underlying invariant probability measure, which is characterized by the
Fredholm integral equation. For this purpose, we also provide an approximated
R\'{e}nyi divergence. Our numerical study shows that this approximated
R\'{e}nyi divergence is reasonably accurate in some simple cases.

The remainder of this paper is organized as follows. In Section~\ref{sec2}, we
define the HMM as a  Markov chain in a Markovian random environment, and
represent the probability as the $L_1$-norm of a product of Markovian random
matrices. Then, we give a brief summary of  eigenvalue and eigenfunction for Markovian
operators. In Section~\ref{sec4}, we study the limiting behavior of the probability and
characterize the R\'{e}nyi divergence. The Kullback--Leibler divergence is defined and can
be regarded as the limit of $\alpha \to 1$ in the R\'{e}nyi divergence. 
In Section~\ref{sec5}, we consider a few examples, including finite state Markov
switching models and RNNs, which are commonly used in machine learning. 
In Section~\ref{sec6}, we  give a numerical computation of the 
R\'{e}nyi divergence by applying the Fredholm integral equation for a two-state
Markov switching model. Section~\ref{sec7} concludes. \\

\def\theequation{2.\arabic{equation}}
\setcounter{equation}{0}
\section{Hidden Markov Models}\label{sec2}

\subsection{Hidden Markov Models}

In this section, we first provide a probability framework for a general HMM under
which it can be regarded as a Markov chain in an enlarged state space. That
is, there are two Markov chains associated with the general HMM to be described
as follows. First, a general HMM is defined as a parameterized Markov chain in
a Markovian random environment with the underlying
environmental Markov chain viewed as missing data. 
Specifically, let ${\bf X}= \{X_n, n \geq 0 \}$ be a Markov chain on a finite
state space ${\cal X}=\{1,\cdots,d\}$, with transition probability  $p_{ij}=
P\{X_1 =j|X_0=i\}$ for $i,j=1,\cdots,d$, and stationary probability $\pi_j$.
Suppose that a random sequence $\{Y_n\}_{n=0}^{\infty},$ taking values in ${\bf
R}^q$, is adjoined to the chain such that $\{(X_n,Y_n), n \geq 0\}$ is a
Markov chain on ${\cal X} \times {\bf R}^q$ satisfying $P \{ X_1 \in A |
X_0=i,Y_0=y \} = P \{ X_1 \in A | X_0=i \}$ for $A \in {\cal B}({\cal X})$, the
$\sigma$-algebra of ${\cal X}$. Conditioning on the full ${\bf X}$
sequence, $Y_n$ is a Markov chain with probability
\begin{eqnarray}\label{2.1}
 P \{Y_{n+1} \in  B | X_0,X_1,\cdots;Y_0,Y_1,\cdots, Y_n \} 
= P \{Y_{n+1} \in  B | X_{n+1},Y_n\}~~~a.s. 
\end{eqnarray}
for each $n$ and $B \in {\cal B}({\bf R}^q),$ the Borel $\sigma$-algebra of
${\bf R}^q$. Note that in (\ref{2.1}) the conditional probability of $Y_{n+1}$
depends on $X_{n+1}$ and $Y_n$ only.
Furthermore, we  assume the existence of the conditional probability density
$f(Y_k|X_k,Y_{k-1})$ of $Y_k$ given  $X_k$ and $Y_{k-1}$ with respect to a
$\sigma$-finite measure ${\cal L}$ on ${\bf R}^q$ such that
\begin{eqnarray}\label{2.2}
 P \{X_1 \in A, Y_{1} \in  B | X_0=i, Y_0 =y_0   \} 
= \sum_{j \in A} \int_{y \in B} p_{ij} f(y|j,y_0)  {\cal L}(dy). 
\end{eqnarray}
We also assume that the Markov chain $\{(X_n,Y_n), n \geq 0\}$ has a stationary
probability with probability density function $\pi_j f(\cdot|j)$ with respect
to ${\cal L}$.
Now we give a formal definition as follows.
\begin{definition}\label{def:hmm}
$\{Y_n,n \geq 0\}$ is called a general hidden Markov model if there is a Markov
chain $\{X_n,n \geq 0\}$ such that the process $\{(X_n,Y_n),n \geq 0\}$ is a
Markov chain satisfying {\rm (\ref{2.1})}. A non-invertible function $g(Y_n)$ of
$Y_n$ is also called a general hidden Markov model. 
\end{definition}

Note that this HMM setting is defined in a general sense, which includes
several interesting examples of Markov-switching Gaussian autoregression, 
cf.\ \cite{Hamilton_1989}, and RNNs in machine learning, 
cf.\ \cite{Goodfellow_Bengio_Courville_2016}. 
When $Y_n$ are {\it conditionally independent} given ${\bf X}$, denote $S_n
= \sum_{t=1}^n Y_t$. 
Then the Markov chain $\{(X_n,S_n),n \geq 0\}$ is called a {\it Markov additive
process}, cf.\  \cite{Ney_Nummelin_1987}, and $\{Y_n,n \geq 1 \}$ is the
celebrated HMM considered in engineering literature. 

Next we follow a similar idea in  \cite{Fuh_2004a} and
\cite{Fuh_Tartakovsky_2019}, to have a Markov chain representation of the
probability. Note that the joint probability of the general HMM $\{Y_n,n \geq
0\}$ is
\begin{eqnarray}\label{joint}
&~& P \{ Y_0 \in B_0, Y_1 \in B_1, \cdots, Y_n \in B_n\} \\
&=&  \int_{y_1 \in B_1} \cdots \int_{y_n \in B_n} p_n(y_1,\cdots,y_n)
{\cal L}(dy_n) \cdots {\cal L}(dy_1), \nonumber 
\end{eqnarray}
where 
\begin{eqnarray}\label{density}
 p_n(y_0,y_1,\cdots,y_n) 
 =\sum_{x_0=1}^d  \cdots \sum_{x_n=1}^d
 \nu_{x_0} f(y_0|x_0)\prod_{k=1}^n p_{x_{k-1}x_k} f(y_k| x_k,y_{k-1}), 
\end{eqnarray}
where  $\nu=(\nu_1,\cdots,\nu_d)^t$ is an initial distribution of $\{X_n, n
\geq 0\}$, which is positive $P$-a.s. Here $t$
denotes the transpose of the underlying vector in ${\bf R}^d$

For a given column vector $a=(a_1,\cdots,a_d)^t \in {\bf R}^d$, define the
$L_1$-norm of $a$ as $\| a \|  = \sum_{k=1}^d |a_k|$.
The probability (\ref{density}) can be represented as
\begin{eqnarray}\label{2.5}
p_n(y_0,y_1,\cdots,y_n) = \| M_n \cdots M_1 M_0 \nu \|,
\end{eqnarray}
where $y_0$ is given  from $f(\cdot|x_0)$ and
\begin{equation}\label{2.6m0}
M_0 = \left[ \begin{array}{cccc}
 f(y_0| x_0=1) & 0 & \cdots &  0 \\
0 & f(y_0| x_0=2) & \cdots  & 0 \\
\vdots & 0 & \ddots  & \vdots\\
0 & \cdots & 0 &  f(y_0| x_0=d)
\end{array} \right], 
\end{equation}
\begin{equation}\label{2.6}
 M_k = \left[ \begin{array}{ccc}
p_{11} f(y_k| x_k=1,y_{k-1}) & \cdots & p_{d1} f(y_k| x_k=1,y_{k-1}) \\
\vdots & \ddots  & \vdots\\
p_{1d}f(y_k| x_k=d,y_{k-1}) & \cdots & p_{dd} f(y_k| x_k=d,y_{k-1})
\end{array} \right], 
\end{equation}
for $k=1,\cdots,n.$

Note that, for $k=1,\cdots,n,$ the quantity $p_{ij}f(Y_k|X_k=j,Y_{k-1})$  in
(\ref{2.6})  represents $X_{k-1} =i$ and $X_k=j$, and $Y_k$ is a Markov chain
with transition probability density $f(y_k|x_k=j,y_{k-1})$ for given ${\bf X}$.
By definition (\ref{2.1}), $\{(X_n,Y_n), n \geq 0 \}$ is a Markov chain, which
implies that $M_k$ is a sequence of Markovian random matrices. Therefore,
by representation (\ref{2.5}), $p_n(Y_0,Y_1,\cdots,Y_n)$ is the $L_1$-norm of a
product of Markovian random matrices. Furthermore, let 
\begin{eqnarray}\label{mirfs}
T_n = M_n \cdots M_1 M_0.
\end{eqnarray} 

Denote $P({\bf R}^d)$ as the projection space on ${\bf R}^d$, $Gl(d,{\bf R})$
as the space of $d \times d$ matrices and $S^{d-1}$ as the $d$-dimensional
sphere.
For $\bar{u} \in P({\bf R}^d)$, $M \in Gl(d,{\bf R})$, let $M\cdot \overline{u}
= \overline{Mu}$, and $\nu=\nu(\theta)= (\nu_1,\cdots,\nu_d)^t \in S^{d-1}$,
the unit sphere with respect to the $L_1$-norm $\|\cdot\|$ in ${\bf R}^d$; we
have
\begin{eqnarray}\label{2.8}
\log \| T_n \nu \| = \log \frac{\| T_n \nu \|}{\| T_{n-1} \nu \|} +
\cdots+\log \frac{\| T_0 \nu \|}{\| \nu \|}.
\end{eqnarray}
Let $Z_n = (X_n,Y_n)$; define
\begin{eqnarray}\label{2.9}
W_0 = ( Z_0, \overline{T_0 \nu} ),~W_1 =(Z_1,\overline{T_1 \nu}), \cdots,
W_n = ( Z_n , \overline{T_n \nu}).
\end{eqnarray}
Then, $W_0,W_1,\cdots,W_n $ is a Markov chain on the state space ${\cal S} :=
{\cal X} \times {\bf R}^q \times P({\bf R}^d)$ with the transition kernel
\begin{eqnarray}\label{2.10}
{\mathbb P}((z,\bar{u}),A \times B)            
:= {\mathbb E}_z(I_{A \times B} (Z_1,\overline{M_1 u}))
\end{eqnarray}
for all $z \in {\cal X} \times {\bf R}^q,~\bar{u} \in P({\bf R}^d),~A \in {\cal
X} \times {\cal B}({\bf R}^q)$, and $B \in {\cal B}(P({\bf R}^d))$, the Borel
$\sigma$-algebra of $P({\bf R}^d)$.

Note that the initial distribution of $W_0$ depends on $Z_0$ only, and $Z_0$
has the distribution $\nu_i f(y|i,y_0)$ as its initial distribution. We note
that ${\mathbb P}_z:={\mathbb P} (\cdot,\cdot)$ in (\ref{2.10}) depends only on
$z$. Let ${\mathbb E}_z:={\mathbb E}_{(z,\bar{u})}$ denote the expectation
under ${\mathbb P}_z$. By (\ref{2.2}), the Markov chain $\{(X_n,Y_n), n \geq
0\}$ has transition density $p_{ij} f(y|j, y_0)$ with respect to $\cal L$.
Therefore, the induced transition probability ${\mathbb P}(\cdot,\cdot)$ has a
probability density $p(\cdot,\cdot)$ with respect to $\cal L$. Under 
condition~C in Section~\ref{sec3}, it follows from Proposition~2 of \cite{Fuh_2004c} or
Proposition 1 of \cite{Fuh_Tartakovsky_2019}  that the Markov   
chain $W_n$ has an invariant probability measure $\Pi$ on
${\cal S}$. Note that ${\cal L}$ is a product measure on ${\cal X} \times {\bf
R}^q \times P({\bf R}^d)$, and the first component has probability density
$\nu_i f(y| i, y_0)$ with respect to $\cal L$.
Now, for $M \in Gl(d,{\bf R})$, let $g : \cal S \times \cal S \rightarrow {\bf
R} $ be $g((z_0,\bar{u}),(z_1,\overline{M u}))= \log \frac{\| M u \|}{\| u
\|};$ then for $\nu$ defined in (\ref{density}),
\begin{eqnarray}\label{2.11}
\log\|T_n \nu \| = g(W_{n-1},W_n)+\cdots+g(W_0,W_1) + g(W_0,W_0)
\end{eqnarray}
is an additive functional of the Markov chain $\{W_n, n \geq 0 \}$,
where $g(W_0,W_0)= \log \frac{\| T_0 \nu \|}{\| \nu \|}.$

\subsection{Nonnegative transition probability  kernel for Markov operator}\label{sec3}

To study the  Kullback--Leibler divergence and R\'{e}yni divergence in  general
HMMs, we must consider a nonnegative transition probability kernel for a
Markov operator of the  induced  Markov chain $\{W_n, n \geq 0\}$ defined in
(\ref{2.9})  on the state space ${\cal S} := {\cal X} \times {\bf R}^q \times
P({\bf R}^d)$, with the transition kernel ${\mathbb P}$ in (\ref{2.10}). Before
that we need the following notation.

Note that $\{(X_n,Y_n),n \geq 0\}$ defined in (\ref{2.1}) and (\ref{2.2}) is a
Markov chain on the state space $ {\cal X} \times {\bf R}^q$. Below,
we abuse the notation a bit to consider  $\{Y_n,n \geq 0\}$ as a Markov chain
on a general state space ${\bf R}^q$. 

\begin{definition}
A Markov chain $\{Y_n,n \ge 0\}$ on a general state space 
${\bf R}^q$  is said to be $V$-uniformly ergodic if there exists a measurable function
$V: {\bf R}^q \rightarrow [1,\infty)$, with $\int V(y) {\cal L}(dy) < \infty$, such that 
\begin{eqnarray}\label{3.11}
~~~ \lim_{n \rightarrow \infty} \sup_{y \in {\bf R}^q} \bigg\{\frac{\big| E[h(Y_n)|Y_0=y]
	- \int h(z){\cal L}(dz)\big|}{V(y)}: |h| \le V \bigg\} =0.
\end{eqnarray}
\end{definition}
\begin{definition}
A Markov chain $\{Y_n,n \geq 0\}$ on a state space ${\bf R}^q$ is said to be  {\it
Harris recurrent}  if there exists a recurrent set ${\cal R} \in {\cal B}({\bf
R}^q)$, a probability measure $\varphi$ on ${\cal R}$, a $\lambda > 0$, and an
integer $n_0$ such that
\begin{eqnarray}\label{mina}
  &~& P\{ Y_n \in {\cal R}~\text{for some}~n \geq 1|Y_0=y\} =1,\\
&~& P\{ Y_{n_0} \in A |Y_0=y \} \geq \lambda \varphi(A), \nonumber
\end{eqnarray}
for all $y \in {\cal R}$ and $A \subset {\cal R}$.
\end{definition}

It is known that under the irreducibility and aperiodicity assumption, $V$-uniform
ergodicity implies that $\{X_n,n \geq 0\}$ is Harris recurrent, 
cf.\ Theorem~9.18 of  \cite{Meyn_Tweedie_2009}.

The following assumptions will be used throughout this paper.

\noindent
Condition C:

\noindent
C1. The Markov chain $\{(X_n,Y_n), n \geq 0 \}$ defined in (\ref{2.1}) and
(\ref{2.2}) is aperiodic and irreducible on ${\cal X} \times {\bf R}^q$. For
each $j \in {\cal X}$, the conditional Markov chain $\{Y_n|X_n, n \geq 0 \}$ is
$V_j(\cdot)$-uniformly ergodic for some $V_j(\cdot)$ on ${\bf R}^q$, such that
there exists $p \geq 1$,
\begin{eqnarray} \label{2.15}
  \sup_{y \in {\cal R}^q } E^{\theta}_{y}\bigg\{\frac{V_j(Y_p)}{V_j(y)} \bigg\} < \infty~\text{for all}~j \in {\cal X}.
\end{eqnarray}

\noindent
C2. Assume 
$0< \sup_{j \in {\cal X}} f(y|j, y_0)< \infty,$ for all $y \in {\bf R}^q$.
Denote $h(Y_1)= \max_{i \in {\cal X}} \sup_{y_{0} \in {\bf R}^q} \\
\sum_{j=1}^d p_{ij} f(Y_1|j,y_0) $. Assume there exists $p \geq 1$ as in C1
such that for all $i \in {\cal X}$,
\begin{eqnarray}
&~& \sup_{j \in {\cal X}, y \in {\bf R}^q} E^{\theta}_{i} \bigg\{ \log \bigg( h(Y_1)^p
\frac{V_j(Y_p)}{V_j(y)}  \bigg) \bigg\}  < 0 \label{2.16},  \\
&~& \sup_{j \in {\cal X}, y \in {\bf R}^q} E^{\theta}_{i} \bigg\{ h(Y_1)
\frac{V_j(Y_1)}{V_j(y)} \bigg\}  < \infty. \label{2.17}
\end{eqnarray}

\noindent
C3. 
Recall that ${\cal L}$ is a
$\sigma$-finite measure on ${\bf R}^q$ defined in (\ref{2.2}). 
Assume
\begin{eqnarray}
\max_{i \in {\cal X}}\sup_{y_0 \in {\bf R}^q } | \sum_{j \in {\cal X}}
\int_{y \in {\bf R}^q} \pi_i  p_{ij} f(y|j,y_0) {\cal L}(dy) |   < \infty. \nonumber
\end{eqnarray}

\begin{remark}\label{rm6}
	 C1 is an ergodic condition for the underlying Markov chain.
The weighted mean contraction property (\ref{2.16}) and the finite weighted
mean average property (\ref{2.17}), which appear in C2, guarantee that the induced
Markov chain $\{W_n, n \geq 0\}$ is $\tilde{V}$-uniformly ergodic for a given function $\tilde{V}$, and hence to be Harris recurrent. In Section~\ref{sec5}, we show that several interesting models satisfy these conditions. 
C3 is a constraint of the R\'{e}nyi divergence (Kullback--Leibler divergence)
and is a standard moment condition.
The finiteness condition is quite natural and holds in most cases.
\end{remark}

The following proposition  is a generalization of Theorem~3 in
\cite{Fuh_2021a}. Since the proof is the same as those in Lemmas~3 and~4 of
\cite{Fuh_2006}, it is omitted.
\begin{proposition}\label{pro1}
	Let $\{(X_n,Y_n), n \geq 0\}$ be the hidden Markov model given in {\rm
	(\ref{2.1})} and {\rm (\ref{2.2})}, satisfying {\rm C1--C3}. 
	Then the induced Markov chain $\{W_n,n \geq 0\}$ is an aperiodic,
	irreducible, and Harris recurrent Markov chain, with the invariant
	probability $\Pi$.	Furthermore there exist $a,C > 0,$ such that ${\mathbb
	E}_{w}(\exp\{a g(W_0,W_1)\}) \leq C < \infty$ for all $w \in {\cal W}.$
\end{proposition}

Under the Harris recurrent condition (\ref{mina}), it is known that, 
cf.\ \cite{Meyn_Tweedie_2009}, $W_n$ admits a regenerative scheme with 
i.i.d.\ inter-regeneration times for an augmented Markov chain, which is called the
``split chain''. Heuristic speaking,  let $\tau_{\Delta}(0) = \tau_{\Delta}$,
and let $\{\tau_{\Delta}(j), j \geq 1\}$ denote the times of consecutive visits
to a recurrent state $\Delta \in {\cal S}$. For a function $f: {\cal S} \to
{\bf R}$, let $S_n =\sum_{i=0}^n f(W_i)$, $S_j(f) =
\sum_{i=\tau_{\Delta}(j)+1}^{\tau_{\Delta}(j+1)} f(W_i)$. 
By the strong Markov property, the random variables $\{S_j(f), j \geq 0\}$ are
independent and identically distributed random variables.  Note that here we
only consider $f(W_i)$; the case of $f(W_i, W_{i+1})$ 
is similar.

Let $\tau = \tau_{\Delta}$ be the first time $(>0)$ reaches the recurrent
state $\Delta$ of the split chain.
Let $\nu$ be an initial distribution on ${\cal S}$, and define
\begin{eqnarray}\label{eigenvalue}
  u(\vartheta,\zeta)= E_{\nu} e^{\vartheta S_\tau -\zeta \tau}~\text{for}~ \zeta \in {\bf R}.
\end{eqnarray} 
Assume that
\begin{eqnarray}\label{mom}
  \Gamma:= \{(\vartheta,\zeta):u(\vartheta,\zeta) < \infty\}~\text{is an open subset on}~{\bf R}^2.
\end{eqnarray}
Denote $\zeta_1:=S_1.$
\cite{Ney_Nummelin_1987} shows that ${\cal D}=\{\vartheta: u(\vartheta,\zeta) <
  \infty~\text{for some}~\zeta\}$
is an open set and that for $\vartheta \in {\cal
D}$, the transition kernel 
\begin{eqnarray}\label{operator}
\hat{\bf P}_{\vartheta}(w,A) = {\mathbb E}_{w}\{ e^{\vartheta \zeta_1} I_{\{W_1 \in A \}}\}
\end{eqnarray} 
has a unique maximal simple real eigenvalue $e^{\Lambda(\vartheta)}$, where
$\Lambda(\vartheta)$ is the unique solution of the equation
$u(\vartheta,\Lambda(\vartheta))=1$, with corresponding right eigenfunctions
$r(\cdot;\vartheta)$ and left eigenmeasures ${\it l}_\nu(\cdot;\vartheta)$
defined by
\begin{eqnarray}\label{eigenfun}
r(w;\vartheta):= {\mathbb E}_{w} \exp\{\vartheta S_\tau - \tau \Lambda(\vartheta)  \}.
\end{eqnarray}
For a measurable subset $A \in {\cal B}({\cal S})$, any initial distribution
$\nu$ on ${\cal S}$ and $w \in {\cal S}$, define
\begin{eqnarray}
{\it l}_\nu(A ;\vartheta) &=&  {\mathbb E}_{\nu}\bigg[\sum_{n=0}^{\tau -1} e^{\vartheta S_n- n \Lambda(\vartheta)}
I_{\{W_n \in A \}}\bigg],  \label{eigenmea1}\\
{\it l}_{w}(A;\vartheta) &=&  {\mathbb E}_{w}\bigg[\sum_{n=0}^{\tau -1} e^{\vartheta S_n- n \Lambda(\vartheta)}
I_{\{W_n \in A \}}\bigg]. \label{eigenmea2}
\end{eqnarray}

To analyze $\hat{{\bf P}}_{\vartheta}$ in (\ref{operator}), for completeness,
we state the following proposition, which is taken from Theorem~4.1 in
\cite{Ney_Nummelin_1987}.
Note that by Proposition~\ref{pro1}, the induced Markov chain $\{W_n,n \geq
0\}$ is an aperiodic, irreducible, and Harris recurrent Markov chain, which
implies that condition~M1 in Theorem~4.1 of \cite{Ney_Nummelin_1987} holds.

\begin{proposition}\label{pro2}
	Let $\{(X_n,Y_n), n \geq 0\}$ be the hidden Markov model given in {\rm
	(\ref{2.1})} and {\rm (\ref{2.2})}, satisfying {\rm C1--C3}. Let $\hat{{\bf
	P}}_{\vartheta} (\cdot, \cdot)$ be the operator  defined in {\rm
	(\ref{operator})}, and $\Lambda(\cdot)$ be defined by the characteristic
	equation {\rm (\ref{eigenvalue})}. Then
	
	(i) ${\cal D}=\{\vartheta: u(\vartheta,\zeta) < \infty~{\rm
	for~some}~\zeta\}$ is an open set. $\Lambda$ is analytic, strictly convex,
	and essentially smooth on ${\cal D}$.
	
	(ii) For $\vartheta \in {\cal D}$, $\lambda(\cdot)=e^{\Lambda(\cdot)}$ is
	the largest eigenvalue of $\hat{\bf P}_\vartheta$ with (right) eigenfunction
	$\{r(w;\vartheta): w \in {\cal S}\}$ and (left) eigenmeasure $\{{\it l}_\nu
	(A; \vartheta): A \in {\cal S}\}$ having the representation {\rm
	(\ref{eigenfun})} and  {\rm (\ref{eigenmea1})}.
	
	(iii) There is a set $B \subset {\cal S}$ with $\varphi(B^c)=0$, such that
	for each $w \in B$, $0 < r(w;\cdot) < \infty$ and is analytic on ${\cal D}$.
	If $B$ is a small set, then
	$0< {\it l}_\nu (B; \vartheta) < \infty$ and $0 < {\it l}_w (B; \vartheta) < \infty$ for all $w \in {\cal S}$ and is analytic on ${\cal D}$.
	
	(iv)  There exists a partition ${\cal S} = \cup_{i=1}^\infty {\cal S}_i$ and a sequence of functions
	$f_i:{\bf R} \to (0,\infty), i =1,2,\cdots,$ such that
	\[ r(w;\vartheta) \geq f_i(\vartheta)I_{{\cal S}_i}(w), ~~~w \in {\cal S}, \vartheta \in {\cal D}, i=1,2,\cdots.\]
\end{proposition}

\begin{remark}
	We will use Proposition~\ref{pro2} (i)--(iv) in Section~\ref{sec4} and
	the rest of this paper,  the reader is referred to \cite{Ney_Nummelin_1987} Theorem 4.1 and Lemma 4.5 for details. (iv) states that there is a countable partition of the state space 
	${\cal S} = \cup_{i=1}^\infty {\cal S}_i$, independent of $\vartheta$, such that $r(w;\vartheta)$
	is uniformly positive on each ${\cal S}_i$.
	However in order to apply  (iii), one needs to extend to
	$0< {\it l}_\nu ({\cal S}; \vartheta) < \infty$ and the uniform boundness of $r(w;\cdot)$ over the whole space ${\cal S}$. To this end, one needs extra condition and apply
    Theorem~4 of \cite{Chan_Lai_2003} under this additional assumption. 
    For completeness, we inculde it as follows.
\end{remark}

Note that $\{W_n, n \geq 0\}$ is $\tilde{V}$-uniformly ergodic as stated in Remark \ref{rm6}.

\noindent
C4. Assume \eqref{mom} hold. Let $C$ be a measurable subset of ${\cal S} $ such that for any given initial distribution $\nu$ on ${\cal S}$,
\begin{eqnarray}\label{eigenpro}
{\cal L}_\nu(C;\vartheta) < \infty~{\rm and}~{\cal L}_{w}(C;\vartheta)< \infty~{\rm for~all~}
w \in {\cal S}.
\end{eqnarray}
Let $\tilde{V}: {\cal S} \to [1,\infty)$ be a measurable function.  Assume for some $0 < \beta < 1$ and $K > 0$,  we have
\begin{eqnarray}
&& E_{w} [e^{\vartheta W_1- \Lambda(\vartheta)} \tilde{V}(W_1)] \leq (1 - \beta) \tilde{V}(w)~\forall~w \notin C, \label{eignmeapro1} \\
&& \sup_{w \in C} {\mathbb E}_{w} [e^{\vartheta W_1- \Lambda(\vartheta)} \tilde{V}(W_1)] = K < \infty~{\rm and}~\int \tilde{V}(w)\varphi(dw) < \infty, ~\label{eignmeapro2}
\end{eqnarray}
where $\varphi$ is defined in \eqref{mina}.
\begin{remark}
Note that under condition C1--C4, we have $0< {\it l}_\nu ({\cal S}; \vartheta) < \infty$ and $r(w;\cdot)$ is uniform boundness over the whole space ${\cal S}$. Althought condition C4
is under the induced Markov chain $\{W_n, n \geq 0\}$, by using the results in \cite{Fuh_2021b},
this condition holds for some interesting examples, see Section \ref{sec5}. Moreover, if  the state space is finite (compact), which is commonely used in engineering, the above results hold.
\end{remark}

\section{ R\'{e}nyi divergence}\label{sec4}
\def\theequation{3.\arabic{equation}}
\setcounter{equation}{0}

We state our main results in this section. Section~\ref{sec4.1} presents the
convergence of the R\'{e}nyi divergence. 
Section~\ref{sec4.2} defines the Kullback--Leibler divergence, and shows that
the Kullback--Leibler divergence  is the  limit of the R\'{e}nyi divergence as
$\alpha \to 1$.

\subsection{R\'{e}nyi divergence}\label{sec4.1}

Let $\{Y_n, n \geq 0\}$ be the general HMM defined in (\ref{2.1}). Denote
$Y_{0:n}=\{Y_0,Y_1,\cdots,Y_n\}$ and $y_{0:n}=\{y_0,y_1,\cdots,y_n\}$. With the
same notation used in Section~\ref{sec2}, denote $P^{(n)}(\cdot)$ and
$Q^{(n)}(\cdot)$ as two probabilities on  $Y_{0:n}$. By (\ref{density}), the
probability density functions $p^{(n)}(\cdot)$ and $q^{(n)}(\cdot)$ of the
random variables $\{Y_0,Y_1,\cdots,Y_n \}$  under $P^{(n)}$ and $Q^{(n)}$ are
given, respectively, by 
\begin{eqnarray}\label{densityP}
p^{(n)} (y_{0:n}) &=& p_n(y_0, y_1,\cdots,y_n) \\
&=& \sum_{x_0=1}^d  \cdots \sum_{x_n=1}^d
\nu_{x_0} f(y_0|x_0)\prod_{k=1}^n p_{x_{k-1}x_k} f(y_k;\theta| x_k,y_{k-1}), \nonumber 
\end{eqnarray}
\begin{eqnarray}\label{densityQ}
q^{(n)} (y_{0:n}) &=& q_n(y_0,y_1,\cdots,y_n) \\
&=& \sum_{x_0=1}^d  \cdots \sum_{x_n=1}^d
\nu_{x_0} g(y_0|x_0)\prod_{k=1}^n p_{x_{k-1}x_k} g(y_k;\theta| x_k,y_{k-1}), \nonumber 
\end{eqnarray}
where $g(y_k| x_k,y_{k-1})$ is the probability density of $Q$ with respect to
${\cal L}$.

Recall the definition of the R\'{e}yni Divergence for independent and
identically distributed random variables (i.i.d.)
$\{\xi_n, n \geq 0\}$ as follows: for given $\alpha \in (0,1) \cup (1,\infty)$, let
\begin{eqnarray}
D_\alpha(f||g) = \frac{1}{\alpha - 1} \log E_f \bigg[ \bigg(\frac{f(\xi_1)}{g(\xi_1)}\bigg)^{\alpha-1} \bigg].
\end{eqnarray}

Now for given a hidden Markov model $\{Y_n, n \geq 0\}$ with transition
probability density $p$ and $q$, let
\begin{eqnarray}\label{renyi}
D^n_\alpha(p^{(n)}||q^{(n)}) &=& \frac{1}{\alpha - 1} \log E_p \bigg[ \bigg(\frac{ p_n(Y_0, Y_1,\cdots,Y_n)}
{ q_n(Y_0, Y_1,\cdots,Y_n)}\bigg)^{\alpha-1} \bigg].
\end{eqnarray}
Note that here $E_f$ ($E_p$) denotes the expectation under probability
distribution $f$ ($p$).
We will use the same type of notation without specification here and afterward.

By (\ref{2.5}) and (\ref{2.8}), we have
\begin{eqnarray}\label{renyimatrix}
(\ref{renyi}) = \frac{1}{\alpha - 1} \log {\mathbb E}_p \bigg[ \bigg(\frac{ \| M^p_n \cdots M^p_1 M^p_0 \nu_p \|}
{ \| M^q_n \cdots M^q_1 M^q_0 \nu_q \|}\bigg)^{\alpha-1} \bigg],
\end{eqnarray}
where $M^p_k$ is defined in (\ref{2.6m0}) and (\ref{2.6}) under the probability
from $p^{(n)}$, and $M_k^q$ is defined in (\ref{2.6m0}) and (\ref{2.6}) under
the probability from $q^{(n)}$, for $k=0,1,\cdots,n$. 
Here $\mathbb E$ is defined as the expectation under the probability $\mathbb P$
defined in (\ref{2.10}).

Denote $g(W_{0},W_0) = \log \frac{\|T_0^p \nu_p \|/\|\nu_p\|} {\|T_0^q \nu_q
\|/\|\nu_q\|}$, and $g(W_{k-1},W_k) = \log \frac{\|T_ k^p \nu_p \|/\|T_{k-1}^p
\nu_p\|}{\|T_ k^q \nu_q \|/\|T_{k-1}^q \nu_q\|}$,  for $k=1,\cdots,n$. Let
$S_0= g(W_0,W_0)$ and  $S_n = S_0 +  \sum_{k=1}^n g(W_{k-1},W_k)$. Then by
(\ref{2.11}), we have
\begin{eqnarray}\label{h(W)}
&~& {\mathbb E}_p \bigg[ \bigg(\frac{ \| M^p_n \cdots M^p_1 M^p_0 \nu_p \|}
{ \| M^q_n \cdots M^q_1 M^q_0 \nu_q \|}\bigg)^{\alpha-1} \bigg] \\
&=& {\mathbb E}_p \bigg[ \exp \bigg\{ \log \bigg(\frac{ \| M^p_n \cdots M^p_1 M^p_0 \nu_p \|}
{ \| M^q_n \cdots M^q_1 M^q_0 \nu_q \|}\bigg)^{\alpha-1} \bigg\} \bigg] \nonumber \\
&=&   {\mathbb E}_p \bigg[ \exp \bigg\{ (\alpha-1)  \bigg( \log \frac{\| T_n^p \nu_p \|/\| T_{n-1}^p \nu_p \|}{ \| T_n^q \nu_q \|/\| T_{n-1}^q \nu_q \|} +
\cdots+ \log \frac{\| T_1^p \nu_p \|/\| T_{0}^p \nu_p \|}{ \| T_1^q \nu_q \|/\| T_{0}^q \nu_q \|} + \log \frac{\| T_0^p \nu_p \|/\|\nu_p\|}{\|T_0^q \nu_q\|/\|\nu_q \|} \bigg) \bigg\} \bigg] \nonumber \\
&=& {\mathbb E}_p \bigg[ \exp \bigg \{(\alpha-1)  \bigg( g(W_{n-1},W_n)+\cdots+g(W_0,W_1) + g(W_0,W_0)
\bigg)   \nonumber  \\
&=& {\mathbb E}_p \bigg[ \exp \big\{ (\alpha-1) S_n\big\} \bigg]. \nonumber
\end{eqnarray}

Then, using $\hat{\bf P}_\alpha$ defined (\ref{operator}) in 
Section~\ref{sec3} with $\vartheta= \alpha - 1$, let $\nu_p$ ($\nu_q$) be the initial
distribution of the $\{W_n, n \geq 0\}$ under $P$ ($Q$). Then we have 
\begin{eqnarray}
&~&  {\mathbb E}_{\nu} e^{(\alpha -1) S_\tau -\Lambda(\alpha) \tau} =	u(\alpha,\Lambda(\alpha))= 1 \label{eigenvalue1} \\ 
&~&  {\mathbb E}_{w} e^{(\alpha -1) S_\tau -\Lambda(\alpha) \tau}= r(w,\alpha).
\label{eigenvalue2}
\end{eqnarray}

Let $\lambda(\alpha)$  be the largest eigenvalue of $\hat{\bf P}_\alpha$.
Denote $r(w,\alpha)$ as the right eigenfunctions associated with
$\lambda(\alpha)$ defined in (\ref{eigenvalue2}). Define
\begin{eqnarray}\label{infsup}
\underline{r}(\alpha) = \inf_w r(w,\alpha),~~~\bar{r}(\alpha)= \sup_w r(w,\alpha), 
\end{eqnarray}
Under conditions C1--C4, by Proposition~\ref{pro2} (iv), the uniform
positivity property, we have  $0 < \underline{r} (\alpha) \leq \bar{r}(\alpha)  < \infty.$

\begin{theorem}\label{thm3.1}
	Under conditions C1--C4, then the R\'{e}nyi divergence rate between $p^{(n)}$
	and $q^{(n)}$ is 
	\begin{eqnarray}\label{3.10}
D_\alpha(p||q):= \lim_{n\to \infty} \frac{1}{n} D^n_\alpha(p^{(n)}||q^{(n)}) = \frac{1}{\alpha - 1} \log \lambda(\alpha),
	\end{eqnarray}
	where $\lambda (\alpha)$  is the largest positive real eigenvalue of
	$\hat{{\bf P}}_\alpha$, and $0 < \alpha < 1$. Furthermore, the same result
	holds for $\alpha > 1$ if $P > 0$ and $Q > 0$.
\end{theorem}

\begin{proof} 
Let $\lambda(\alpha)$ be the largest positive real eigenvalue of $\hat{{\bf
P}}_\alpha$ defined in (\ref{operator}), with associated positive right
eigenfuction $r(w,\alpha) > 0$ uniformly on ${\cal S}$. 
Then by (\ref{eigenvalue2}), we have
\begin{eqnarray}\label{4.16B}
\nu_\alpha \hat{{\bf P}}_\alpha^{n-1} r(w;\alpha) = \lambda^{n-1} (\alpha) r(w;\alpha).
\end{eqnarray}
Let $\underline{r} (\alpha)$ and  $\bar{r}(\alpha)$ be defined in (\ref{infsup}).
Then $0 < \underline{r} (\alpha) \leq r(w,\alpha) \leq \bar{r}(\alpha)  <
\infty,$  for all $w \in \cal S$. 
Let $\nu_\alpha \hat{{\bf P}}_\alpha^{n-1} {\bf 1} =  b(\alpha)$. Then by
(\ref{infsup}), we have 
\begin{eqnarray}
  \lambda^{n-1} (\alpha) r(w;\alpha) = \nu_\alpha \hat{{\bf P}}_\alpha^{n-1} r(w;\alpha) \leq \bar{r}(\alpha) b(\alpha). \nonumber
\end{eqnarray}
Similary we have $\lambda^{n-1} (\alpha) r(w;\alpha) \geq \underline{r}(\alpha)
b(\alpha).$

Therefore,
\begin{eqnarray}
\frac{r(w;\alpha)}{\bar{r}(\alpha)} \leq \frac{b(\alpha)}{ \lambda^{n-1} (\alpha)}  \leq \frac{r(w,\alpha)}{\underline{r}(\alpha)}. \nonumber
\end{eqnarray}

Since 
\begin{eqnarray}
\frac{\int_w r(w;\alpha) \Pi(dw)}{\bar{r}(\alpha)} \leq \frac{b(\alpha) }{ \lambda^{n-1} (\alpha)}  \leq \frac{\int_w r(w;\alpha) \Pi(dw)}{\underline{r}(\alpha)}, \nonumber
\end{eqnarray}
we have
\begin{eqnarray}\label{upperlower}
\frac{1}{n} \log \frac{\int_w r(w;\alpha) \Pi(dw)}{\bar{r}(\alpha)} \leq \frac{1}{n} \log \frac{b(\alpha) }{ \lambda^{n-1} (\alpha)}  \leq \frac{1}{n} \log \frac{\int_w r(w;\alpha) \Pi(dw)}{\underline{r}(\alpha)}.
\end{eqnarray}
Note that the constant terms in the upper- and lower-bound in
(\ref{upperlower}) are bounded and independent of $n$, which approach $0$ as
$n \to \infty$. Therefore, we have
   \begin{eqnarray}
    \lim_{n \to \infty} \frac{1}{n} \log \frac{\nu_\alpha \hat{{\bf P}}_\alpha^{n-1} {\bf 1} }{ \lambda^{n-1} (\alpha)}  = 0.
   \end{eqnarray}
   
   Hence
    \begin{eqnarray}\label{3.14}
  \lim_{n \to \infty} \frac{1}{n} \log \nu_\alpha \hat{{\bf P}}_\alpha^{n-1} {\bf 1}  
   =  \lim_{n \to \infty} \frac{1}{n} \log \lambda^{n-1} (\alpha)+ \lim_{n \to \infty} \frac{1}{n} \log \frac{\nu_\alpha
   	\hat{{\bf P}}_\alpha^{n-1} {\bf 1} }{ \lambda^{n-1} (\alpha)} = \log \lambda(\alpha).
   \end{eqnarray}
   Note that (\ref{3.14}) holds for both $p$ and $q$.
   Thus
   \begin{eqnarray}
\lim_{n \to \infty} \frac{1}{n} D^n_\alpha(p^{(n)}||q^{(n)}) 
   = \frac{1}{\alpha-1} \log \lambda(\alpha).
   \end{eqnarray}
   The proof is complete. 
   \end{proof}
   
\subsection{Kullback--Leibler divergence}\label{sec4.2}
   
By making use of Theorem~\ref{thm3.1}, we herein show that the R\'{e}nyi
divergence reduces to the Kullback--Leibler divergence as $\alpha \to 1$. 
Let us first note the following result about the computation of the
Kullback--Leibler divergence rate between two general HMMs. The convergence rate
of the Kullback--Leibler divergence has been investigated by \cite{Fuh_2004a},
and \cite{Fuh_Mei_2015} for the parametric case. In the following proposition, 
we then show that the Kullback--Leibler  divergence for general HMMs can also be
written in a form similar to that in the i.i.d.\ case.

Recall that $\{Y_0,Y_1,\cdots,\}$ is a general HMM. Let $p^{(n)}$ and $q^{(n)}$
be two probability distributions. Let $P$ and $Q$ be the probabilities associated
with $p^{(n)}$ and $q^{(n)}$, respectively. Let $\nu_p$ and $\nu_q$ be two
initial distributions with respect to $p^{(n)}$ and $q^{(n)}$, respectively. If
$Q > 0$, then $q > 0$. Denote
\begin{eqnarray}\label{kl}
D^n(p^{(n)}||q^{(n)}) &=& \log  \bigg(\frac{ p_n(Y_0, Y_1,\cdots,Y_n)}
{ q_n(Y_0, Y_1,\cdots,Y_n)}\bigg).
\end{eqnarray}
\begin{proposition}\label{Proposition1}
Under conditions C1--C4, the Kullback--Leibler divergence rate between
$p^{(n)}$ and $q^{(n)}$ is well-defined with
\begin{eqnarray}\label{eqn:KLrep}
K(p, q) =  \lim_{n \to \infty}  \frac{1}{n} D^n(p^{(n)}||q^{(n)}) = {\mathbb E_{\Pi,p}} \left[ \log p_1(Y_{0},Y_1) \right]
 -  {\mathbb E_{\Pi,q}}\left[ \log q_1(Y_0, Y_{1}) \right],
\end{eqnarray}
where ${\mathbb E_{\Pi,p}}$ (${\mathbb E_{\Pi,q}}$) is the expectation of
${\mathbb P}_p$ (${\mathbb P}_q$) defined in (\ref{2.10})  of Section~\ref{2.1}
under the invariant probability $\Pi$ of $\{W_n,n \geq 0\}$. Here ${\mathbb P}_p$ denotes
the probability under $p$.
\end{proposition}
\begin{proof} 
 Under conditions C1--C4, by Proposition~\ref{pro2}, the invariant probability
 $\Pi$ of the induced Markov chain $\{W_{n}, n \geq 0\}$ exists. Recall
 $T_n=M_n \cdots M_1 M_0$ defined in (\ref{mirfs}). Now let $M_n^p$ ($M_n^q$)
 be $M_n$ defined in (\ref{2.6m0}) and (\ref{2.6}) when the probability is
 under $P$ ($Q$). We can define $T_n^p$ and $T_n^q$ similarly. First, it is
 easy to see from (\ref{2.11}) that
	\begin{align}\label{pfKLexist-1}
	\notag
	& \frac{1}{n} \left[ \log p_n(Y_0,Y_1,\cdots,Y_n) - \log p_n(Y_0,Y_1,\cdots,Y_n) \right] 
	= \frac{1}{n} \left[ \log \|T_n^p \| - \log \|T_n^q\| \right] \\
	= & \frac{1}{n} \sum_{i=1}^n g_p(W_{i}, W_{i-1}) - g_q(W_{i}, W_{i-1}).
	\end{align}
	Taking $n \rightarrow \infty$ on both sides of \eqref{pfKLexist-1}, then by
	Proposition~\ref{pro1} and the SLLN for Markov random walks in
	\cite{Meyn_Tweedie_2009}, we have
	\begin{align*}
	K(p, q)  = & {\mathbb E_{\Pi,p}} \left[ g_p(W_1, W_0) \right]
	- {\mathbb E_{\Pi,q}} \left[ g_q(W_1, W_0) \right] \\
	= & {\mathbb E_{\Pi,p}} \left[ \log p_1(Y_{0},Y_1) \right]
	-  {\mathbb E_{\Pi,q}}\left[ \log q_1(Y_0, Y_{1}) \right],
	\end{align*}
	which completes the proof. 
	\end{proof}

\begin{theorem}
	Let $\alpha \in (0,1) \cup (1,\infty)$. 
Assume conditions C1--C4 hold; then 
	\begin{eqnarray}\label{3.18}
	&~& \lim_{\alpha \to 1} \lim_{n\to \infty}  \frac{1}{n} D^n_\alpha(p^{(n)}||q^{(n)}) =  \lim_{n\to \infty} \lim_{\alpha \to 1}  \frac{1}{n} D^n_\alpha(p^{(n)}||q^{(n)})   \\
	&=& \int_{w_0} \int_{w_1} \Pi_{w_0} {\mathbb P}(w_0,w_1) \log \frac{{\mathbb P}(w_0,w_1)}{{\mathbb Q}(w_0,w_1)}  
	d w_1 d w_0 = K(p,q), \nonumber
	\end{eqnarray} 
	which is the Kullback--Leibler divergence defined in {\rm (\ref{eqn:KLrep}). }
\end{theorem}

\begin{proof}	To prove (\ref{3.18}), we first consider the case of $
\lim_{n\to \infty} \lim_{\alpha \to 1}  \frac{1}{n}
D^n_\alpha(p^{(n)}||q^{(n)})$.  By (\ref{renyi}), (\ref{renyimatrix}), and
(\ref{h(W)}), we have
\begin{eqnarray}\label{h(W)1}
D^n_\alpha(p^{(n)}||q^{(n)}) &=& \frac{1}{\alpha - 1} \log E_p \bigg[ \bigg(\frac{ p_n(Y_0, Y_1,\cdots,Y_n)}
{ q_n(Y_0, Y_1,\cdots,Y_n)}\bigg)^{\alpha-1} \bigg] \\
&=& \frac{1}{\alpha - 1} \log {\mathbb E}_p \bigg[ \bigg(\frac{ \| M^p_n \cdots M^p_1 M^p_0 \nu_p \|}
{ \| M^q_n \cdots M^q_1 M^q_0 \nu_q \|}\bigg)^{\alpha-1} \bigg] \nonumber \\
&=& \frac{1}{\alpha - 1} \log {\mathbb E}_p \bigg[ \exp \big\{ (\alpha-1) S_n\big\} \bigg]. \nonumber
\end{eqnarray}

Then
\begin{eqnarray}\label{h(W)2}
&~&   \lim_{n\to \infty} \lim_{\alpha \to 1}  \frac{1}{n} D^n_\alpha(p^{(n)}||q^{(n)})   \\
&=& \lim_{n\to \infty}  \frac{1}{n}  \lim_{\alpha \to 1} \frac{1}{\alpha - 1}   \log {\mathbb E}_p \bigg[ \exp \big\{ (\alpha-1) S_n \big\} \bigg]. \nonumber \\
&=& \lim_{n\to \infty}  \frac{1}{n}   {\mathbb E}_p \big[  S_n  \big]
= \int_{w_0} \int_{w_1} \Pi_{w_0} {\mathbb P}(w_0,w_1) \log \frac{{\mathbb P}(w_0,w_1)}{{\mathbb Q}(w_0,w_1)}  
d w_1 d w_0. \nonumber
\end{eqnarray}
Note that the second identity comes from L'Hospital's Rule and the last
identity in (\ref{h(W)2}) comes from  (\ref{eqn:KLrep}) in 
Proposition~\ref{Proposition1}.

Next, we consider the case of $\lim_{\alpha \to 1} \lim_{n\to \infty}
\frac{1}{n} D^n_\alpha(p^{(n)}||q^{(n)})$.

By (\ref{3.10}) in Theorem~\ref{thm3.1}, we have  
\[ \lim_{n\to \infty} \frac{1}{n} D^n_\alpha(p^{(n)}||q^{(n)}) = \frac{1}{\alpha - 1} \log \lambda (\alpha). \]
To evaluate $\lim_{\alpha \to 1}\frac{1}{\alpha - 1} \log \lambda (\alpha)$,
note by Proposition~\ref{pro2}~(i) that the eigenvalue $\lambda(\alpha)$ of
$\hat{\bf P}_\alpha$  is a continuous differentiable function of $\alpha$.
Note that since $Q > 0$, we have 
\[ \lim_{\alpha \to 1} \lambda(\alpha) = 1. \]
Let $a$ denote an arbitrary base of the logarithm. Then, by L'Hopital's rule,
we find that
\begin{eqnarray}\label{3.21}
\lim_{\alpha \to 1} \frac{\log \lambda(\alpha)}{\alpha -1} = \frac{1}{\ln a} \lambda'(1) := \frac{1}{\ln a} \frac{\partial \lambda(\alpha)}{\partial \alpha}\bigg|_{\alpha=1},
\end{eqnarray}
which is well defined by Proposition~\ref{pro2} since the algebraic
multiplicity of $\lambda(\alpha)$ is $1$ by 
(\ref{operator}). 
The equation defining the largest positive eigenvalue $\lambda(\alpha) = 1$ of
$\hat{\bf P} (w,A):= {\mathbb E}_w \{I_{W_1 \in A}\} = {\mathbb P}_w(A)$. By
Proposition~\ref{pro2}~(i), $\hat{\bf P}_\alpha$ is analytic for $\alpha \in
{\cal D}$; therefore by (\ref{operator}), it is straightforward to check that
$\hat{\bf P}_\alpha \to \hat{\bf P}$ as $\alpha \to 1$.

Note that  for $\zeta \in {\bf R}$,  $u(\alpha,\zeta)= {\mathbb E}_{\Pi}
e^{(\alpha-1) S_\tau -\zeta \tau}$ defined in (\ref{eigenvalue}). Then the
transition kernel $\hat{\bf P}_{\alpha}(w,A ) = {\mathbb E}_{w}\{ e^{(\alpha-1)
\zeta_1} I_{\{W_1 \in A \}}\}$ has a maximal simple real eigenvalue
$\lambda(\alpha) = e^{\Lambda(\alpha)}$, where $\Lambda(\alpha)$ is the unique
solution of the equation $u(\alpha,\Lambda(\alpha))=1$. Then using $\Lambda(1)
= \log \lambda(1) = \log 1 =0$, we have
\begin{eqnarray}\label{3.22}
&~& u(\alpha,\Lambda(\alpha))=1 \Longrightarrow \frac{\partial u(\alpha,\Lambda(\alpha))} {\partial \alpha}|_{\alpha=1}=0 
\Longrightarrow {\mathbb E}_{\Pi}  \{(S_\tau - \Lambda'(\alpha) \tau) e^{(\alpha-1) S_\tau -\Lambda(\alpha) \tau} \}|_{\alpha=1}=0 \nonumber \\
&\Longrightarrow& {\mathbb E}_{\Pi}  \{(S_\tau - \Lambda'(\alpha) \tau) e^{(\alpha-1) S_\tau -\Lambda(\alpha) \tau} \}|_{\alpha=1}=0
\Longrightarrow {\mathbb E}_{\Pi}  \{(S_\tau - \Lambda'(1) \tau)  \}=0 \nonumber \\
&\Longrightarrow& \lambda'(1) = \frac{\partial \lambda(\alpha)}{\partial \alpha}|_{\alpha =1}= \frac{  {\mathbb E}_{\Pi}  S_\tau}{{\mathbb E}_{\Pi} \tau} = {\mathbb E}_\Pi S_1.
\end{eqnarray}
The last identity in (\ref{3.22}) comes from Lemma 5.2 of \cite{Ney_Nummelin_1987}.

By using (\ref{3.21}) and (\ref{3.22}), we  obtain
\begin{eqnarray}\label{3.19a}
 \lim_{\alpha \to 1} \frac{1}{\alpha -1} \log \lambda(\alpha)   
= \int_{w_0} \int_{w_1} \Pi_{w_0} {\mathbb P}(w_0,w_1) \log \frac{{\mathbb P}(w_0,w_1)}{{\mathbb Q}(w_0,w_1)}  d w_1 d w_0 = K(p,q), 
\end{eqnarray}
which completes the proof.  
\end{proof}

\def\theequation{4.\arabic{equation}}
\setcounter{equation}{0}
\section{Examples}\label{sec5}


We present two examples of general HMMs in this section. Section~\ref{msm}
considers the Markov switching models, whereas Section~\ref{RNN} studies the
RNN.

\subsection{Markov switching models}\label{msm}

We start with a simple real valued $q$-order autoregression
around one of $d$ constants $\mu_1,\cdots, \mu_d$:
\begin{eqnarray}\label{mg1}
Y_n - \mu_{X_n}
= \sum_{k=1}^q \psi_k(Y_{n-k} - \mu_{X_{n-k}}) + \varepsilon_n,
\end{eqnarray}
where $\varepsilon_n \sim N(0,\sigma^2)$, $|\psi_k|<1$ for $k=1,\cdots,q$, and
$\{X_n, n \geq 0\}$ is a $d$-state ergodic Markov chain. When $q=4$ and $d=2$,
this model was studied by Hamilton (1989) in order to analyze the behavior of
the U.S.\ real GNP.
Note that the Markov switching model (\ref{mg1}) includes the classical HMM by
letting $\psi_k=0$ for $k=1,\cdots,q$. To apply our theory in the form of
(\ref{mg1}), we consider a simple case of order~$1$ in (\ref{mg1}) with $d=2$.
The extension to the general case is straighforward.
In this case, the conditional probability given $X_n = x_n$ and
$Y_{n-1}=y_{n-1}$, $n\geq 1$, is
\begin{eqnarray}\label{mg2}
f(y_n|x_n, y_{n-1};\theta) = \frac{1}{\sqrt{2\pi}\sigma} \exp \bigg( -[ (y_n -
\mu_{x_n}) -  \psi_1(y_{n-1} - \mu_{x_{n-1}})]^2/2\sigma^2 \bigg).
\end{eqnarray}
Denote $[p_{ij}]_{i,j=1,2}$ as the transition probability of the underlying
Markov chain $\{X_n, n \geq 0\}$ and let $\theta =(p_{11},
p_{21},\psi_1,\mu_1,\mu_2,\sigma^2)$ be the given parameter.
Assume that $| \psi_1 | <1$ for the stability property, and that there exists a
constant $c>0$ such that $\sigma^2 >c$. Moreover, we assume that $\mu_1 \neq
\mu_2$. 
Since the state space of $X_n$ is finite, we consider $0< p_{ij} < 1$ for all
$i,j =1,2$, and for $j=1,2$ let $V_j(y)= |y| +1 $ 
(cf.\ page~394 of \cite{Meyn_Tweedie_2009})
such that the condition C1 holds. 
Under the normal distribution assumption, it is easy to see that (\ref{2.15})
in  conditions C1  and C3 holds.

Next we check that the mean contraction property (\ref{2.16}) in C2 
for  a simple Markov switching model with general innovation holds. 
	Given $p \geq 1$ as in C2, and $|\psi| < 1$, let  $X_n$ be a two-state
	Markov chain, and $Y_n = \mu_{X_n} + \psi Y_{n-1} + \varepsilon_n,$ where
	$\varepsilon_n$ are i.i.d.\ random variables with $E|\varepsilon_1|=a <
	\infty$. Further, we assume both $\varepsilon_1$  have a positive probability
	density function with respect to the Lebesgue measure.
	Denote $h(Y_1)=C < 1$, $b=(1-|\psi|^p)/(1-|\psi|)$ and choose $p$ such that $C^p (ab + 1) < 1$. 
	Let $d(u,v)=|u-v|.$ 
	Then we have 
	\begin{eqnarray}\label{4.3b}
	&~& \sup_y \bigg\{{E}_j\bigg(\log \frac{h(Y_1)^p V_j(Y_p)}{V_j(y)}|Y_0=y\bigg)\bigg\} \\
	&<& \sup_y \bigg\{{E}_j\bigg(\log \frac{C^p(|\psi^p y + \sum_{k=0}^{p-1} \psi^k 
		\varepsilon_{p-k}|+1)}{|y|+1}|Y_0=y\bigg)\bigg\} \nonumber \\
	&<& \log   \sup_{y} \bigg\{  \frac{C^p ( |\psi^p y| + E|\sum_{k=0}^{p-1} \psi^k \varepsilon_{p-k}|+1)}{|y|+1} \bigg\}
	= \log   \sup_{y}\bigg\{ \frac{C^p (|\psi^p y| +  ab + 1)}{|y|+1} \bigg\} < 0. \nonumber
	\end{eqnarray}
	By using the same argument, it is straightforward to check that
	(\ref{2.17}) in C2 holds. By making use the same argement as that in Example 2 of \cite{Chan_Lai_2003}, we can check C4 hold.
	 In Section~\ref{sec7}, we will present a numerical
	computation method of the R\'{e}nyi divergence under model (\ref{mg1}) with
	$q=1$ and $d=2$.

\subsection{Recurrent neural network}\label{RNN}

Note that at the outset, RNN is a non-linear dynamical system commonly trained
to fit sequence data via some variant of gradient descent. In this subsection, 
we treat RNN as a stochastic model as usual, 
cf.\ \cite{Goodfellow_Bengio_Courville_2016}. \cite{Fuh_2021b} considers the example
of RNN from a Markovian-iterated function system point of view, to study its
stability. Here we apply a general HMM point of view to investigate RNN.
Although some notation overlaps between these two parts, we include it
here for completeness.
A connection between the classical finite state HMM and RNN is in \cite{Buys_Bisk_Choi_2018}.

An RNN can take as input a variable-length sequence $y=(y_1, \cdots,y_n)$ by
recursively processing each symbol while maintaining its internal hidden state
$h$. At each time step $n$, the RNN reads the symbol $Y_n \in {\bf R}^q$ and
updates its hidden state $h_n \in {\bf R}^p$ by
\begin{eqnarray}\label{RNN1}
h_n = f_{\theta} (Y_n,h_{n-1}),
\end{eqnarray}
where $f_\theta$ is a deterministic non-linear transition function, and
$\theta$ is the parameter of $f_\theta$. 

The transition function $f_\theta$ can be implemented with gated activation
functions such as long short-term memory (LSTM) or the gated recurrent unit (GRU).
The joint probability of the RNN model sequence can be written as a product
of conditional probabilities such that
\begin{eqnarray}\label{RNN2}
P(Y_1, \cdots,Y_n) = \prod_{k=1}^n P(Y_k|Y_1,\cdots,Y_{k-1}) = \prod_{k=1}^n g_{\gamma}(h_{n-1}), 
\end{eqnarray}
where $g_\gamma$ is a function that maps the RNN hidden state $h_{t-1}$ to a
probability distribution over possible outputs, and $\gamma$ is the parameter
of $g_\gamma$. 

To analyze (\ref{RNN1}) and (\ref{RNN2}), we provide a Markov chain framework
as follows.  
Specifically, let ${\bf H}= \{h_n, n \geq 0 \}$ be a sequence of random
variables on $({\bf R}^p, {\cal B({\bf R}}^p))$, and suppose that a random
sequence $\{Y_n\}_{n=0}^{\infty}$ taking values in ${\bf R}^q$ is adjoined to
${\bf H}$ such that $\{Z_n:=(h_{n-1}, h_n,Y_n), n \geq 0\}$  is a Markov chain
on ${\bf R}^p \times {\bf R}^p \times {\bf R}^q$ satisfying 
\begin{eqnarray}\label{RNN3a}
&~& P \{(h_{n-1},h_n) \in A, Y_{n} \in  B | h_0,h_1,\cdots,h_{n-1};Y_0,Y_1,\cdots, Y_{n-1} \} \\
&=& \int_{y \in B} P\{ (h_{n-1},h_n) \in A | h_0,h_1,\cdots,h_{n-1};Y_0,Y_1,\cdots, Y_{n-1},  Y_{n} \in  dy \} \nonumber \\
&~& ~~~~~~\times P\{ Y_{n} \in  dy | h_0,h_1,\cdots,h_{n-1};Y_0,Y_1,\cdots, Y_{n-1} \} Q(dy) \nonumber \\
&=& \int_{y \in B} P\{ (h_{n-1},h_n) \in A | (h_{n-2},h_{n-1}); Y_{n} \in  dy \}   P \{ Y_{n} \in  dy | (h_{n-2},h_{n-1}); Y_{n-1} \} Q(dy) \nonumber \\
&=& \int_{y \in B} I_{ \{ (h_{n-1},h_n) \in A | (h_{n-2},h_{n-1}); Y_{n} \in  dy \}}   P \{ Y_{n} \in  dy |( h_{n-2},h_{n-1}); Y_{n-1} \} Q(dy) \nonumber 
\end{eqnarray}
for  $A \in {\cal B}({\bf R}^p \times {\bf R}^p)$, $B \in {\cal B}({\bf R}^q)$
and  each $n=1,2,\cdots$.

By considering $X_n$ to be degenerate and $\{(h_{n-1},h_n,Y_n), n \geq 0\}$ to be a
general-state Markov chain, we have $Y_n=g((h_{n-1},h_n,Y_n))$ as a general
HMM. 
Next, we  consider the simple case, 
cf.\ \cite{Chung_Kastner_Dinh_Goel_Courville_Bengio_2015}, in which the generating
distribution is conditioned on $h_{n-1}$ such that
\begin{eqnarray}\label{6.8a}
Y_n = \mu_{y,n} + \sigma_{y,n}  \varepsilon_n,
\end{eqnarray}
where   $\sigma_{y,n} > 0$ P-a.s., $\varepsilon_n \sim N(0,1)$ is a sequence of
i.i.d.\ random variables, and $\varepsilon_n$ is independent of $\{Y_{n-k}, k
\geq 1\}$ for all $n$. Here, we assume that $\mu_{y,n}$ and $\sigma^2_{y,n}$
are the parameters of the generating distribution such that $(\mu_{y,n},
\sigma^2_{y,n})\sim g_{\gamma}(h_{n-1}), $ with $g_{\gamma}$ any highly
flexible function such as a neural network.

To illustrate the general HMM approach, we consider two examples: linear RNN
and LSTM.
To start with, we consider the  linear RNN. Let $Y_n$ be the output model in
(\ref{6.8a}); the linear RNN updates its hidden state using the following
recurrence equation:
\begin{eqnarray}\label{RNN5a}
h_n = f_\theta( Y_n,  h_{n-1}) = \delta_0 +  \delta_1 h_{n-1} +  \delta_2 Y_{n},
\end{eqnarray}
where $\theta= (\delta_0,\delta_1,\delta_2)$ with $\delta_0 > 0,~ \delta_1 >
0,$  and $\delta_2 > 0$ constants. 


For an explicit representation, we analyze  the linear RNN (\ref{6.8a}) and
(\ref{RNN5a}) as follows:
let $Z_n=(h_{n-1}, h_{n},Y_{n})$ be the Markov chain on ${\cal X}:=({\bf R}
\times {\bf R} \times {\bf R})$. Denote $\eta_n = h_{n-1}^{-1} Y_{n}$ and let
$\tau_n = (\delta_1 + \delta_2 \eta_n) \in {\bf R}$.  Let $A_n$ be a $ 3 \times
3$ matrix written  as	
\begin{equation}\label{6.10}
A_{n} = \left[ \begin{array}{ccc}
0 & 1 & 0 \\
0 & \tau_{n}& 0 \\
0 & \eta_{n} & 0 \\ \end{array} \right].
\end{equation}

Note that although $\{A_n, n \geq 0 \}$ are random matrices driven by the
Markov chain $\{Z_n, n \geq 0\}$, since the randomness of $Y_n$ comes
only from the i.i.d.\ random variables $\varepsilon_n$, and since $Y_n$ is independent
of ${\cal F}_{n-1}$, the $\sigma$-algebra generated by $Y_1,\cdots,Y_{n-1}$,
$\{A_n, n \geq 0 \}$ are i.i.d.\ random matrices. 

Let ${\xi}_n=(0,\delta_0,0)^t \in {\bf R}^{3}$.
Then we have the following linear state space representation of the linear RNN
(\ref{6.8a}) and (\ref{RNN5a}): $Z_n$ is a Markov chain governed by
\begin{equation}\label{6.11}
Z_{n} = A_{n} Z_{n-1} + \xi_{n},
\end{equation}
and $Y_n := g(Z_n)$, the observed random quantity, is a non-invertible function
of $Z_n$. 

It is easy to check that the stability condition holds if ${\mathbb E}_{\Pi}
\frac{Y_1}{h_{1}} < 1$ and $\delta_1 + \delta_2  {\mathbb E}_{\Pi}
\frac{Y_1}{h_{1}} < 1$, where $\Pi$ is the stationary distribution of the
Markov chain $\{(Z_n, A_n \cdots A_1), n \geq 0\}$. 
The moment conditions  hold under the normality assumption in (\ref{6.8a}).
Note that $Z_n$ defined in (\ref{6.11}) is a $V$-uniformly ergodic Markov chain
with $V(z)=\|z\|^2$, cf.\ Theorem~16.5.1 of \cite{Meyn_Tweedie_2009}. By using
the results in \cite{Bougerol_Picard_1992}, it is straightforward to check that the
stability condition and conditions C1 and~C3 hold.
By an argument similar to that in (\ref{4.3b}), C2 holds.

Next, we consider the  LSTM network, cf.\ \cite{Hochreiter_Schmidhuber_1997}. By
using (\ref{6.8a}) as the output model for $Y_n$, we consider the hidden unit
as follows. The state is a pair of vectors $s = (c, h) \in {\bf R}^{2d}$, and
the model is parameterized by eight matrices, 
$W_{\triangle} \in {\bf R}^{2d}$ and $U_{\triangle} \in {\bf R}^{d \times n}$,
for $\triangle \in \{i,f,\sigma,z\}$. The state-transition map $\phi_{\mathit{LSTM}}$
for $f_\theta$ in (\ref{RNN5a}) is defined as 
\begin{eqnarray}
&~& f_t = \sigma(W_f h_{t-1} + U_f y_t),~~~~~~ i_t = \sigma(W_i h_{t-1} + U_i y_t), ~~~o_t = \sigma(W_o h_{t-1} + U_o y_t), 
\label{LSTM} \\
&~& z_t = \tanh(W_z h_{t-1} + U_z y_t),~~~ c_t = i_t \circ z_t + f_t \circ c_{t-1}, ~~~ h_t =o_t \cdot \tanh(c_t), \nonumber
\end{eqnarray}
where $\circ$ denotes elementwise multiplication, and $\sigma$ is the logistic
function.

Let $Z_n=(h_{n-1}, h_{n},Y_{n})$ be the Markov chain defined in  (\ref{6.8a})
and (\ref{RNN5a}).
	To provide conditions under which the $r$-step iterated system
	$\phi^r_{\mathit{LSTM}} = \phi_{\mathit{LSTM}} \circ \cdots \circ  \phi_{\mathit{LSTM}}$ is stable, we
	denote $\|W\|_\infty$ as the induced $\ell_\infty$ matrix norm, which
	corresponds to the maximum absolute row sum $\max_i \sum_j |W_{ij}|$, and
	let $E\|f\|_\infty = \sup_t E\|f_t\|_\infty$. 
	Since $\sigma < 1$ for given any weights $W_f;~U_f$ and inputs $y_t$, we
	have $E\|f\|_\infty < 1$. 
	This means the next state $c_t$ must ``forget'' a non-trivial portion of
	$c_{t-1}$. 
	We leverage this phenomenon to give sufficient conditions for $\phi_{\mathit{LSTM}}$
	to be
	contractive in the $\ell_\infty$ norm, which in turn implies the  system
	$\phi_{\mathit{LSTM}}^r$
	is contractive in the $\ell_2$ norm for
	$r = O(\log d)$. 

	By using the mean contraction under the normal distribution defined in
	(\ref{6.8a}), the following result is taken from Proposition~1 of
	\cite{Fuh_2021b}, in which he shows that the iterated function system
	$\phi_{\mathit{LSTM}}^r$ is stable; see also Proposition~2 in
	\cite{Miller_Hardt_2019} for deterministic LSTM. 
	\begin{proposition}\label{pr1}
		$\|W_f\|_\infty < B_w< \infty$, $\|U_f\|_\infty < B_u< \infty$,
		$\|Y_t\|_\infty \leq B_Y,~P$-$a.s$ for some random variable $B_Y$ with $E
		B_Y< \infty$. Moreover, assume
		$\|W_i\|_\infty < (1 - E\|f\|_\infty)$, $\|W_o\|_\infty < (1 - E\|f\|_\infty)$, 
		$\|W_z\|_\infty < (1/4)(1 - E\|f\|_\infty)$, $\|W_f\|_\infty < (1 - E\|f\|_\infty)^2$,
		and $r = O(\log d)$; then the iterated function system 
		$\phi_{\mathit{LSTM}}^r$
		is stable.
	\end{proposition}

Under this assumption, the state space of the Markov chain
$\{Z_n=(h_{n-1},h_n,Y_n), n \geq 0\}$ defined in the  LSTM model  (\ref{LSTM})
is compact, 
and hence is Harris recurrent and satisfies the $V$-uniformly ergodic assumptions C1. 
Under the normality assumption in  (\ref{6.8a}), it is easy to see that $E
|\varepsilon_1|^{p} < \infty$  for any $p >0$. Therefore the moment conditions
of C1 and C3 hold. 
The contraction property C2 holds due to the definition of the activation
functions. C4 holds as the state space of the Markov chain is compact.
By Theorem~\ref{thm3.1} and Proposition~\ref{Proposition1}, 
we prove the existence of the R\'{e}nyi divergence and Kullback--Leibler
divergence, and provide a characterization.
As numerical computations seem difficult, we will rely on Monte-Carlo
simulations.

\section{Computational Issues in General HMM}\label{sec6}
\def\theequation{5.\arabic{equation}}
\setcounter{equation}{0}

Since the R\'{e}nyi divergence in general HMM involves an eigenvalue  which is
difficult to compute, we provide an approximated
R\'{e}nyi divergence in Section~\ref{KL}. Next we present a theoretical
background of the invariant measure for $J_p^\alpha$ in Section~\ref{5.2},
and report numerical computation of the R\'{e}nyi divergence in Section~\ref{Sim}.

\subsection{Approximated R\'{e}nyi Divergence}\label{KL}

Under conditions C1--C4, by Theorem~\ref{thm3.1} the R\'{e}nyi divergence
between $P$ and $Q$ is 
\begin{eqnarray}\label{RenyiDef}
D_\alpha(p||q):= \lim_{n\to \infty} \frac{1}{n} D^n_\alpha(p^{(n)}||q^{(n)}) = \frac{1}{\alpha - 1} \log \lambda (\alpha).
\end{eqnarray}

Note that the computation of the R\'{e}nyi divergence based on (\ref{RenyiDef})
involves the computation of the largest eigenvalue $\lambda (\alpha)$, which
is not an easy task. Hence, instead of using (\ref{RenyiDef}), we will provide
an alternative approach based on the recursive formula as follows.
Recall that from (\ref{renyi}), $D_\alpha(p||q)$ is defined as 
\begin{eqnarray}\label{renyi2}
D_\alpha(p||q):= \lim_{n \to \infty} \frac{1}{n} D^n_\alpha(p^{(n)}||q^{(n)}) &=& \frac{1}{\alpha - 1} \lim_{n \to \infty} \frac{1}{n}  \log {\mathbb E}_p \bigg[ \bigg(\frac{ p_n(Y_0, Y_1,\cdots,Y_n)}{ q_n(Y_0, Y_1,\cdots,Y_n)}\bigg)^{\alpha-1} \bigg],
\end{eqnarray}
where $ {\mathbb E}_p $ denotes the expectation according to $ {\mathbb P}_p $,
the probability of the Markov chain $\{W_n, \geq 0\}$ when the probability of
$\{Y_0,Y_1,\cdots, Y_n \}$ is under $p_n$.

Now, we seek to show that (\ref{renyi2}) can be computed via iterations. In
particular, we have
\begin{eqnarray}\label{5.2a}
p_n(y_0,y_1,\ldots,y_n)=C_{n,1}(p) + \ldots+C_{n,d}(p),
\end{eqnarray}
where 
\begin{eqnarray}\label{5.2b}
C_{t,j}(p)=f(y_t|x_{t-1}= j,y_{t-1})\sum_{s=1}^dp_{sj} C_{t-1,s}(p),
\end{eqnarray}
with initial values $C_{1,j}(p)=\pi_j(p)f(y_1|x_0=j)$, for $j=1,\ldots,d$. The
following lemma
gives an iteration method to compute the probability.
\begin{lemma}\label{L2}
Under conditions C1--C4, the probability can be calculated by
	\begin{align}\label{MSMLogLL}
	(p_n (y_0,y_1,\dots,y_n))^{\alpha-1}= \prod_{t=1}^n (C^*_{t,1}(p)+\ldots+C^*_{t,d}(p))^{\alpha-1},
	\end{align}
	where
	$$C^*_{t,j}(p)=\frac{f(y_t|x_{t-1}= j,y_{t-1})\sum_{s=1}^dp_{sj} C_{t-1,s}(p)}{\sum_{s=1}^dC_{t-1,s}(p)},$$ 
	with initial values $C^*_{1,j}(p)=\pi_j(p)f(y_1|x_0=j)$, for $j=1,\ldots,d$.
\end{lemma}
\begin{proof} 
By (\ref{5.2a}) and (\ref{5.2b}), and using induction for $n=1,2,\cdots,$ it is
easy to show that
\begin{align*}
&\prod_{t=1}^n (C^*_{t,1}(p)+\ldots+C^*_{t,d}(p))^{\alpha-1} \\
=& \exp \big\{ (\alpha-1)\sum_{t=1}^n \log (C^*_{t,1}(p)+\ldots+C^*_{t,d}(p)) \big\} \\
=& \exp \bigg\{ (\alpha-1) \bigg[ \sum_{t=1}^n   \log \big( \sum_{j=1}^d  C_{t,j}(p) \big) -  \sum_{t=1}^n \log \big( \sum_{s=1}^dC_{t-1,s}(p)\big) \bigg] \bigg\} \\
=& \exp \bigg\{ (\alpha-1)    \log \big( \sum_{j=1}^d  C_{n,j}(p) \big)  \bigg\} =   
 \big( \sum_{j=1}^d  C_{n,j}(p) \big)^{\alpha-1 }  = (p_n (y_0,y_1,\dots,y_n))^{\alpha-1}.
\end{align*}
\end{proof}
\begin{remark}
Note that $C^*_{t,j}(p)$ in (\ref{MSMLogLL}) generally have
finite means, and thus $C^*_{n,j}(p)$ increase or decrease linearly in $n$.
Hence, Lemma~\ref{L2} provides an algorithm that is computationally feasible
when the time step is large. In addition, it is easy to see from
(\ref{MSMLogLL}) that the ``normalization'' of (\ref{5.2a}) reflects the idea of
using the projection space $P({\bm R}^d)$ defined in (\ref{2.9}) and
(\ref{2.10}).
\end{remark}
	
	Denote
	\begin{eqnarray}\label{jalpha}
	J^\alpha
	= {\mathbb E}_{\Pi}\Big[{\mathbb E}_{\Pi}\bigg(\frac{C_{t,0}^*(p)+C_{t,1}^*(p)} {C_{t,0}^*(q)+C_{t,1}^*(q)}\bigg)^{\alpha-1} \big| X_{t-1},Y_{t-2},W_{t-1} \big) \Big].
	\end{eqnarray}
	By (\ref{renyi2}), (\ref{MSMLogLL}), and (\ref{jalpha}) we observe that the R\'{e}nyi
	divergence can be approximated as	
  \begin{eqnarray}\label{RenyiLim}
	 D_\alpha(p||q) 
	&=& \frac{1}{\alpha-1}   \lim_{n\to\infty} \frac{1}{n} \log {\mathbb E}_p \Bigg\{  \frac{ \exp \big\{ (\alpha-1) \sum_{t=0}^n \log\left(\sum_{s=1}^d C^*_{t,s}(p)\right) \big\} }{ \exp \big\{ (\alpha-1) \sum_{t=0}^n \log\left(\sum_{s=1}^dC^*_{t,s}(q)\right) \big\}} \Bigg\},  \\
	&=& \frac{1}{\alpha-1}  \frac{1}{n}  \log \lim_{n\to\infty} {\mathbb E}_p \bigg\{  \exp\bigg\{  \sum_{t=0}^n \log \left( \frac{\sum_{s=1}^d C^*_{t,s}(p)}{\sum_{s=1}^dC^*_{t,s}(q)}   \right)^{\alpha-1}   \bigg\} \bigg\} \nonumber \\
	&\approx& \frac{1}{\alpha-1} \frac{1}{n}  \log  \big(J^\alpha \big)^n = \frac{1}{\alpha-1}   \log  J^\alpha.
	  \nonumber 
	\end{eqnarray}

	\begin{remark}\label{lm5.3}
		Note that here in (\ref{RenyiLim}), we apply the following approximation
		\begin{eqnarray}\label{jp}
	{\mathbb E}_p \bigg\{  \exp\bigg\{ \frac{1}{n} \sum_{t=0}^n \log \left( \frac{\sum_{s=1}^d C^*_{t,s}(p)}{\sum_{s=1}^dC^*_{t,s}(q)}   \right)^{\alpha-1}   \bigg\} \bigg\}
		\approx  J^\alpha.
		\end{eqnarray}
		In other words, we approximate the largest eigenvalue $\lambda(\alpha)$
		via $\lambda^n(\alpha) \approx  (J^\alpha)^n$, which can be explained as
		follows.
		By (\ref{eigenvalue}), we have for~given~$\tau=n$,
		\[	 E_{\nu} \bigg[e^{(\alpha -1) S_n} | \tau=n  \bigg]= E_\nu \bigg[{ e^{\Lambda(\alpha)}}^{n}| \tau=n\bigg] = E_\nu \bigg[\lambda^n(\alpha) | \tau=n \bigg]. \]
		
		Let $T^\alpha_1$ be defined as (\ref{mirfs}) with the form in (\ref{renyi2}). Denote $A_\tau$ as the number of epochs by the regeration time $\tau$; then use $E\tau_1 \cdot A_\tau
		\sim n$ to  approximate 
		\[ (E(T^\alpha_1))^n \approx Ee^{(\alpha -1) S_n} \approx E [e^{(\alpha -1) \sum_{j=1}^{\tau_n} S_{\tau_j}} ] \approx E [ \prod_{j=1}^{\tau_n}  e^{(\alpha - 1) S_{\tau_j}} ]
		\approx  E [ e^{(\alpha -1) S_{\tau_1}} ]^{\tau_n} = \lambda^{E\tau_1 \cdot A_\tau}(\alpha) \approx \lambda^n(\alpha). \]
		In summary, we have more accurate approximation when $\alpha \approx 1$
		or in the `almost i.i.d.' case. 
		In other words, we approximate $\lambda(\alpha)$ by $E(T^\alpha_1)$ via the idea of
		approximating the transition probability by the invariant probability. 
	\end{remark}
	
\subsection{Theoretical Background of the Invariant Measure for $J^\alpha$}\label{5.2}
Since the R\'{e}nyi divergence in general HMM in Theorem~\ref{thm3.1}
involves the largest eigenvalue of the operator defined in the induced Markov
chain $\{W_n, n \geq 0\}$, it is not easy to compute in general. One
standard way to calculate the R\'{e}nyi divergence in general HMM is via Monte
Carlo simulation. Specifically, we first generate $\{Y_t\}_{t=1}^n$ from the
model ${\mathbb P}_p$. Second, we compute
\[	\frac{1}{n}  \log {\mathbb E}_p \bigg[ \bigg(\frac{ p_n(Y_0, Y_1,\cdots,Y_n)}{ q_n(Y_0, Y_1,\cdots,Y_n)}\bigg)^{\alpha-1} \bigg] \]
via (\ref{MSMLogLL}). Then, the R\'{e}nyi divergence can be estimated by
repeating the above procedure several times and averaging its results. Needless to
say, Monte Carlo is time-consuming especially when repeated calculations are needed. In this section, we propose a faster
algorithm to compute the R\'{e}nyi divergence in a two-state Markov switching
model. By (\ref{RenyiLim}), the R\'{e}nyi divergence can be computed as
$D_\alpha(p||q)=\frac{1}{\alpha-1}\log J^{\alpha}$,
where $J^\alpha$ is defined in (\ref{jalpha}). For the case ${\cal X
}=\{0,1\}$, $J^{\alpha}$ can be computed numerically. Before stating the
method, we first note that the invariant measure of $J^{\alpha}$ depends
only on $W_t=C^*_{t,0}(p)/(C^*_{t,0}(p)+C^*_{t,1}(p))\in[0,1]$. The other
important fact is that $W_t$ depends to $X_t$ and $Y_{t-1}$ due to the fact
that $\{(X_t,Y_t,Y_{t-1},W_t), t \geq 0 \}$ is a Markov chain ($Y_{-1}:=0$). To
find the stationary distribution of $W_t$, we define $m_j(\cdot,\cdot)$ as the
stationary density function satisfying
$$\text{Pr}(X_t=j, Y_{t-1}=u,W_t\leq x)=\int_{0}^x m_j(u, w)dw.$$
By extending the argument as in \cite{Fuh_Mei_2015}, the following
proposition characterizes $m_j(\cdot,\cdot)$ via Fredholm integral equations.
Before that, we require the following notation. 

For ease of presentation, we will use $p_{\bm \theta_1}$ for probability under
${\mathbb P}_p$
	and $p_{\bm \theta}$ for a probability under  ${\mathbb P}$ with a parameter
	$\bm \theta$. We will denote
$\varphi_j(\bm \theta)$ as the parameter of the underlying Markov chain $X_t=j$
with parameter $\bm \theta$.
	
Define $z(w,x)$ as
$$z(w,x)=\frac{x}{1-x}\cdot\frac{p_{01}(\bm\theta) w+p_{11}(\bm\theta)(1-w)}{p_{00}(\bm\theta)w+p_{10}(\bm\theta) (1-w)};$$
for $j=0,1$,
$$Q_j(u,z)= P\phantom{}_{\bm \theta_1}\Bigg( \frac{g(Y_t|\varphi_0({\bm \theta}),u)}{g(Y_t| \varphi_1({\bm \theta}),u)}\leq z\Bigg| X_t=j, Y_{t-1}=u\Bigg).$$

By using an argument similar to Theorem~3 of \cite{Fuh_Mei_2015}, we have
\begin{proposition}\label{thm1}
Under conditions C1--C4, for all $0< x< 1$,
	\begin{align}\label{M00}
	m_0(u,x)&=p_{00}(\bm\theta_1)\int_0^1\int_{-\infty}^\infty g(u|\varphi_0({\bm \theta_1}),v)\frac{\partial }{\partial x}Q_0(u,z(w,x))m_0(v,w)dvdw \nonumber \\
	&\qquad+p_{10}(\bm\theta_1)\int_0^1\int_{-\infty}^\infty g(u|\varphi_1({\bm \theta_1}),v)\frac{\partial }{\partial x}Q_0(u,z(w,x))m_1(v,w)dvdw, \nonumber \\ 
	m_1(u,x)&=p_{01}(\bm\theta_1)\int_0^1\int_{-\infty}^\infty g(u|\varphi_0({\bm \theta_1}),v)\frac{\partial }{\partial x}Q_{1}(u,z(w,x))m_0(v,w)dvdw \nonumber \\
	&\qquad+p_{11}(\bm\theta_1)\int_0^1\int_{-\infty}^\infty g(u|\varphi_1({\bm \theta_1}),v)\frac{\partial }{\partial x}Q_{1}(u,z(w,x))m_1(v,w)dvdw.
	\end{align}
\end{proposition}

\begin{remark}
	The key observation is that for a Markov switching model with a finite
	number of $d$-hidden states, the invariant measure can essentially be
	defined by  $d$-functions whose ranges are in the $(d-1)$-dimensional space.
	This is computationally challenging for $d \geq 3$, and numerically
	computationally feasible for $d=2$, as the corresponding two-dimensional
	real-valued functions can be characterized by a two-dimensional Fredholm
	integral equation. Note that the Fredholm integral equation is well studied in
	mathematics, and the two-dimensional case can be numerically solved by
	discretizing and then finding the eigenvector of a (large) square matrix
	with respect to the eigenvalue. 
	
\end{remark}

To illustrate the usefulness of Proposition~\ref{thm1}, we consider the
following Markov switching regression model, in which the mean depends on
$X_t$. A more general case of the means $\mu_{X_t}$ and $\mu_{X_{t-1}}$ is in
Section~\ref{Sim} for the numerical computation.

\begin{example}
Let	$\{X_n, n \geq 0\}$ be a two-state ergodic (aperiodic, irreducible, and
positive recurrent) Markov chain with transition probability matrix 
${P}_{\bm \theta}=\begin{pmatrix}
p_{00} & p_{01} \\
p_{10} &  p_{11}
\end{pmatrix}.$ Denote
	\begin{eqnarray}\label{5.6}
	Y_t=\mu_{X_t}+\psi_{X_t}Y_{t-1}+\sigma_{X_t}\epsilon_t,
	\end{eqnarray}
	where $\epsilon_t\sim N(0,1)$. Denote ${ \bm \theta} =(p_{00}, p_{11},
	\mu_0, \mu_1,  \psi_0, \psi_1, \sigma_0,\sigma_1)$. 
	To compute $Q_j(u,z)$ in (\ref{5.6}), note that
	$$\frac{g(Y_t|\varphi_0({\bm\theta}),u)}{g(Y_t|\varphi_1({\bm\theta}),u)}=\frac{\sigma_1}{\sigma_0}\exp\left\{ \zeta \left[ \left( Y_t+\frac{\eta}{\zeta} \right)^2 -\frac{\eta^2}{\zeta^2} +\frac{\nu}{\zeta} \right] \right\},$$
	where
	$$\zeta=\frac{1}{2\sigma_1^2}-\frac{1}{2\sigma_0^2},
	\ 
	\eta=\frac{\mu_0+\psi_0 u}{2\sigma_0^2}-\frac{\mu_1+\psi_1 u}{2\sigma_1^2}
	\ \text{and}\ 
	\nu=-\frac{(\mu_0+\psi_0u)^2}{2\sigma_0^2}+\frac{(\mu_1+\psi_1u)^2}{2\sigma_1^2}.$$
	It follows that $Q_j(u,z)$ can be computed as
	$$Q_j(u,z)=\Pr \Bigg( \Xi\leq \frac{1}{\zeta}\log\left(\frac{\sigma_0}{\sigma_1}z\right)+\frac{\eta^2}{\zeta^2}-\frac{\nu}{\zeta}\Bigg),$$
	where $\Xi\sim\sigma_j^2\chi^2_1((\mu_j+\psi_ju+\eta/\zeta)^2).$

Once we can approximate the stationary density $m_j(\cdot,\cdot)$ in
Proposition~\ref{thm1}, $J^\alpha$ can be computed as follows and the R\'{e}nyi
divergence can be estimated.

\begin{align}\label{Jalpha}
J^\alpha
&= \mathbb{E}_{\Pi}\Big[ \mathbb{E}_{\Pi}\big(\exp(\log \bigg(\frac{C_{t,0}^*(\bm\theta_1)+C_{t,1}^*(\bm\theta_1)}{C_{t,0}^*(\bm\theta)+C_{t,1}^*(\bm\theta)}\bigg)^{\alpha-1}) \big| X_{t-1},Y_{t-2},W_{t-1} \big) \Big]\\
&=  \mathbb{E}_{\Pi}\Big[ \mathbb{E}_{\Pi}\bigg(\frac{C_{t,0}^*(\bm\theta_1)+C_{t,1}^*(\bm\theta_1)}{C_{t,0}^*(\bm\theta)+C_{t,1}^*(\bm\theta)}\bigg)^{\alpha-1} \big| X_{t-1},Y_{t-2},W_{t-1} \big) \Big] \nonumber \\
&=  \sum_{j=0}^1\int_0^1\int_{-\infty}^\infty\sum_{i=0}^1 \mathbb{P}_{\text{inv}}(X_t=j|X_{t-1}=i)
\nonumber \\
&\qquad\cdot\mathbb{E}_{\Pi}\bigg(\frac{C_{t,0}^*(\bm\theta_1)+C_{t,1}^*(\bm\theta_1)}{C_{t,0}^*(\bm\theta)+C_{t,1}^*(\bm\theta)}\bigg)^{\alpha-1} \big| X_t=j, X_{t-1}=i, Y_{t-2}=v, W_{t-1}=w \big) m_i(v,w) dvdw \nonumber \\
&= \int_0^1\int_{-\infty}^\infty\Big[p_{00}(\bm\theta_1)G^\alpha_{00}(v,w)+p_{01}(\bm\theta_1)G^\alpha_{01}(v,w)\Big]
m_0(v,w) dvdw \nonumber \\
&\qquad+ \int_0^1\int_{-\infty}^\infty\Big[p_{10}(\bm\theta_1)G^\alpha_{10}(v,w)+p_{11}(\bm\theta_1)G^\alpha_{11}(v,w)\Big]m_1(v,w) dvdw, \nonumber
\end{align}
where 
\begin{align}\label{Gij}
&~G^\alpha_{ij}(v,w)=\text{E}_{\Pi}\bigg(\frac{C_{t,0}^*(\bm\theta_1)+C_{t,1}^*(\bm\theta_1)}{C_{t,0}^*(\bm\theta)+C_{t,1}^*(\bm\theta)}\bigg)^{\alpha-1} \big| X_t=j, X_{t-1}=i, Y_{t-2}=v, W_{t-1}=w \big)\\
&=\int_{-\infty}^\infty\int_{-\infty}^\infty  
\Big( \frac{[p_{00}(\bm\theta_1)w+p_{10}(\bm\theta_1)(1-w)]g(y|\varphi_0(\bm\theta_1),u)  + [p_{01}(\bm\theta_1) w+p_{11}(\bm\theta_1)(1-w)]g(y|\varphi_1(\bm\theta_1),u)}{[p_{00}(\bm\theta)w+p_{10}(\bm\theta)(1-w)]g(y|\varphi_0(\bm\theta),u) + [p_{01}(\bm\theta) w+p_{11}(\bm\theta)(1-w)]g(y|\varphi_1(\bm\theta),u)} \Big)^{\alpha-1} \nonumber \\ 
&~~~\cdot g(u|\varphi_i(\bm\theta_1),v)g(y|\varphi_j(\bm\theta_1),u)dudy. \nonumber
\end{align}

Then  
\begin{eqnarray}\label{Renyi}
D_\alpha(p||q)	= \frac{1}{\alpha-1}   \log J^\alpha
=\frac{1}{\alpha-1}   \log J^\alpha.
\end{eqnarray}
\begin{remark}
	We can see that this method involves two parts. The first part is solving
	the eigenvalue to approximate $m_j$, and the second part is using double
	integration to approximate $J^\alpha$. If we carefully design our algorithm,
	we need only a few seconds to calculate the R\'{e}nyi divergence.
\end{remark}
\end{example}

\subsection{Numerical Computation of the R\'{e}nyi Divergence}\label{Sim}

To illustrate our method, in this subsection, we consider the following Markov
switching model: 
\begin{eqnarray}\label{example}
Y_t=\psi_1 \mu_{X_t}+\psi_2 \mu_{X_{t-1}}+\phi Y_{t-1}+\epsilon_t, 
\end{eqnarray}
where $\epsilon_t\sim N(0,\sigma^2)$ and $X_t\in {\cal X} =\{0,1\}$ is a Markov
chain with transition probability matrix 
${P}_\theta=\begin{pmatrix}
p_{00} & p_{01} \\
p_{10} & p_{11}
\end{pmatrix}.$
Denote ${ \bm \theta} =(p_{01}, p_{10}, {\bm \mu}, \phi, \psi_1, \psi_2,
\sigma)$. The derivation of the Markov switching model (\ref{example}) for
numerical study will be given in the Appendix.

We give a summary of our numerical approximation of the invariant measure $\Pi$
in model (\ref{Jalpha}) and (\ref{Gij}). Even though there are several ways to
solve the two-dimensional Fredholm integral equation, the algorithm for solving
the above integral equation should be carefully designed because $\bm 0$ is one
of its solutions. To resolve this problem, we vectorize $m_j$ and take grid
points on $(0,1)\times(-\infty,\infty)$. 
Note that the standard techniques for the general Fredholm integral equation
can be used to estimate $m_0(\cdot,\cdot), m_1(\cdot,\cdot)$ numerically from
equation~(\ref{M00}) in Proposition~\ref{thm1}. The key step is to consider two
discrete approximations:   
$\int_{0}^{1} \int_{-\infty}^{\infty} g(u|\phi(\bm \theta),v) \frac{\partial}{\partial x}
Q(u, z(w,x)) m(v,w) dv dw$ and $\frac{\partial}{\partial x} Q(u, z(w,x)).$
To this end, we first approximate $\int_{-\infty}^{\infty}$ by $\int_{-a}^a$,
and consider the case that $-a=v_0 < v_1< \cdots < v_{N-1} < v_{N} = a$, $0=x_0
< x_1< \cdots < x_{N-1} < x_{N} = 1$ and $0 = w_0 < w_1 < \cdots <w_{N-1} <
w_{N} = 1$. One choice is to simply set $v_i=\frac{2ai}{N}$ and $x_i = w_i =
\frac{i}{N}$ for $i=0,1,\ldots, N.$
Assume $m(a,0) = m(a,1)=m(b,0)=m(b,1)=0.$
Then
\begin{eqnarray*}
&~&	 \int_{0}^{1} \int_{-a}^{a} g(u|\phi(\bm \theta),v) \frac{\partial}{\partial x} Q(u, z(w,x)) m(v,w) dv dw \\
&\approx&
	\Big[\frac{1}{N} g(u|\phi(\bm \theta),v_1) \frac{\partial}{\partial x} Q(u, z(w_1,x)) \Big] m(v_1,w_1) \\
	&~& + 
	 \sum_{j=2}^{N-1}  \frac{1}{2N}\Big[ g(u|\phi(\bm \theta),v_j) \frac{\partial}{\partial x} Q(u, z(w_j,x)) +g(u|\phi(\bm \theta),v_j)  \frac{\partial}{\partial x} Q( u,z(w_{j+1},x)) \Big] m(v_j,w_{j})
\end{eqnarray*}
and
\begin{eqnarray*}
	\frac{\partial}{\partial x} Q( u,z(u,x)) &\approx&  \frac{Q( u,z(u,x + \Delta)) - Q( u,z(u,x - \Delta))}{2 \Delta},
\end{eqnarray*}
for $x=x_1, \ldots, x_i, \ldots, x_{N-1},$ where $\Delta = 1 / (2N).$

Now we apply $(u,x)=(u_1,x_1), \ldots, (u_{N-1},x_{N-1})$ to 
equation~(\ref{M00}) to obtain $2 (N-1)$ equations, and use the above approximations to
discretize the right-hand side of equation~(\ref{M00}). Write these $2 (N-1)$
equations in a matrix form to yield ${\bf w} = {\bf M} {\bf w},$ where ${\bf w}
= ( m_0(u_1,x_1), \cdots, m_0(u_{N-1},x_{N-1}), m_1(u_1,x_1), \cdots,\\
m_1(u_{N-1},x_{N-1}))^{t}$ and ${\bf M}$ is a $2(N-1) \times 2(N-1)$ matrix
whose entries values depend on $Q_{0}(u_i,z(w_{j}, x_{i} \pm \Delta))$ and
$Q_{1}(u_i,z(w_{j}, x_{i} \pm \Delta)),$ both of which can be computed for
given values of $(u_i,x_{i})$ and $w_{j}.$ Since $Q_{i}(u,z(w,1)) = 1$ and
$Q_{i}(u,z(w,0)) = 0$ for all $w$ and $i=0,1,$ a nice property of the matrix
${\bf M}$ is that the sum of each column equals $1,$ which is the main reason
why we use the  invariant joint density notation. This property ensures that
the matrix ${\bf M}$ has an eigenvalue (the largest) equal to $1$ and thus the
corresponding eigenvector {\bf w} is an efficient approximation to
$m_0(\cdot,\cdot)$ and $m_1(\cdot,\cdot).$

With the discretized approximation $\hat m_0(\cdot,\cdot)$ and $\hat
m_1(\cdot,\cdot),$ we can estimate $J_{\theta}$ by approximating the
integration in (\ref{Jalpha}), thereby yielding an alternative way to compute
the R\'{e}nyi divergence.
Now we show this alternative way is valuable, as the only method so far in the
literature to estimate the R\'{e}nyi divergence in general HMM  is the Monte
Carlo simulation of $\frac{1}{n} \log S_{n}$ for large value of $n.$ It is
expected that our proposed alternative non-Monte-Carlo method will allow one to
check the accuracy and correctness of both methods.



\begin{table}
	\setlength{\tabcolsep}{4pt}
	\renewcommand{\arraystretch}{0.8}
	\centering
	\caption{\label{tb:examples}Parameters for the eight examples.}
		\vspace{-0.5cm}
	\begin{tabular}{lll}
		\toprule 
    (1) & $\bm\theta = (0.41, 0.6, (1,0), 0, 1, 0, 2)$ & $\bm\theta_1 = (0.41, 0.6, (2,1), 0, 1, 0, 1.5)$\\\cmidrule(lr){1-3}

	(2) & $\bm\theta = (0.41, 0.59, (1,0), 0, 1, 0, 2)$ & $\bm\theta_1 = (0.41, 0.59, (2,1), 0, 1, 0, 1.6)$\\\cmidrule(lr){1-3}
	
    (3) & $\bm\theta = (0.4, 0.59, (1,0), 0, 1, 0, 1)$ & $\bm\theta_1 = (0.4, 0.59, [2,1], 0, 1, 0, 0.9)$\\\cmidrule(lr){1-3}

    (4) & $\bm\theta = (0.4, 0.599, (1,0), 0, 1, 0, 1)$ & $\bm\theta_1 = (0.4, 0.599, (2,1), 0, 1, 0, 0.9)$\\\cmidrule(lr){1-3}

    (5) & $\bm\theta = (0.59, 0.4, (1,0), 0, 1, 0, 1)$ & $\bm\theta_1 = (0.59, 0.4, (2,1), 0, 1, 0, 0.9)$\\\cmidrule(lr){1-3}

	(6) & $\bm\theta = (0.599, 0.4, (1,0), 0.2, 1, 0, 1)$ & $\bm\theta_1 = (0.599, 0.4, (2,1), 0.3, 1, 0, 1.1)$\\\cmidrule(lr){1-3}

    (7) & $\bm\theta = (0.4, 0.59, (1,0), 0.2, 1, 0.2, 1.1)$ & $\bm\theta_1 = (0.4, 0.59, (2,1), 0.1, 1, 0.1, 1)$\\\cmidrule(lr){1-3}

    (8) & $\bm\theta = (0.4, 0.59, (1,1), 0, 1, 0, 1)$ & $\bm\theta_1 = (0.4, 0.59, (2,2), 0, 1, 0, 0.9)$\\
    \bottomrule
	\end{tabular}
\end{table}

We focus on the following 8 cases listed in Table~\ref{tb:examples}, and
Table~\ref{T1} presents the R\'{e}nyi divergence for various $\alpha=0.5,0.8,
0.99, 0.999, 1.001, 1.01, 1.5, 2$, and  the Kullback--Leibler divergence
($\alpha\rightarrow 1$) for the above 8 cases with two methods: simulation and
numerical approximation. For simulation, we use the sample size as 2000
and the replication number as 100; for numerical approximation, we use
the discretized lattice number as 16 and the lower and upper integral bounds
as 15.

Next, we also report the value of the R\'{e}nyi divergence  and
Kullback--Leibler divergence in this context. We consider two ways to estimate
the R\'{e}nyi divergence  and Kullback--Leibler divergence. One is based on
Monte Carlo simulations with the time step for convenience. The other
is based on the invariant probability measure based on 500 discretizations over
the interval [0,1]. The corresponding results are summarized in Table~\ref{T1}, in which
we report R\'{e}nyi divergence based on numerical and Monte Carlo methods. The number
in the parentheses indicates the standard deviation of the Monte Carlo estimator. We also
present relative errors between these two methods and their corresponding running time in 
seconds.

\begin{table}
	\setlength{\tabcolsep}{4pt}
	\renewcommand{\arraystretch}{0.8}
	\centering
	\caption{\label{T1}Results for cases with various $\alpha$}
		\vspace{-0.5cm}
	\begin{tabular}{llccccccccc}
		\toprule 
		& & \multicolumn{8}{c}{Cases}& Average\\\cmidrule(lr){3-10}
		$\alpha$ & &1 & 2 & 3 & 4 & 5 & 6 & 7 & 8 &  time (sec) \\ \midrule
		\multirow{4}{*}{0.5} & Numerical & 0.1091 & 0.0921 & 0.2196 & 0.2211 & 0.2225 & 0.2723 & 0.1227 & 0.2818 & 11.60 \\ \cdashline{2-11} 
		& \multirow{2}{*}{Simulation} & 0.1097 & 0.0927 & 0.2211 & 0.2220 & 0.2239 & 0.2733 & 0.1293 & 0.2826 & 239.45 \\
		&  & (0.0145) & (0.0133) & (0.0214) & (0.0214) & (0.0217) & (0.026) & (0.015) & (0.0247) &  \\\cdashline{2-11} 
		& R.E. (\%) & -0.5469 & -0.6472 & -0.6784 & -0.4054 & -0.6253 & -0.3659 & -5.1044 & -0.2831 &  \\ \midrule
		\multirow{4}{*}{0.8} & Numerical & 0.1533 & 0.1324 & 0.3372 & 0.3387 & 0.3366 & 0.4566 & 0.1939 & 0.4243 & 11.68 \\ \cdashline{2-11} 
		& \multirow{2}{*}{Simulation} & 0.1538 & 0.1329 & 0.3382 & 0.3395 & 0.3374 & 0.4575 & 0.1979 & 0.4250 & 240.43 \\ 
		& & (0.0114) & (0.0109) & (0.0191) & (0.0192) & (0.0189) & (0.0275) & (0.0131) & (0.0212) &  \\ \cdashline{2-11} 
		& R.E. (\%) & -0.3251 & -0.3762 & -0.2957 & -0.2356 & -0.2371 & -0.1967 & -2.0212 & -0.1647 &  \\ \midrule
		\multirow{4}{*}{0.99} & Numerical & 0.1762 & 0.1541 & 0.4072 & 0.4087 & 0.4032 & 0.5850 & 0.2363 & 0.5062 & 12.39 \\ \cdashline{2-11} 
		& \multirow{2}{*}{Simulation} & 0.1767 & 0.1546 &  0.4079 & 0.4094 & 0.4036 & 0.5857 & 0.2388 & 0.5068 & 238.79 \\
		&&  (0.0101) & (0.0099) & (0.0184) & (0.0185) & (0.018) & (0.0295) & (0.0124) & (0.0199) &  \\ \cdashline{2-11} 
		& R.E. (\%) & -0.2830 & -0.3234 & -0.1716 & -0.1710 & -0.0991 & -0.1195 & -1.0469 & -0.1184 &  \\ \midrule
		\multirow{4}{*}{0.999} & Numerical & 0.1772 & 0.1550 &  0.4104 & 0.4120 & 0.4063 & 0.5913 & 0.2382 & 0.5099 & 12.31 \\ \cdashline{2-11} 
		& \multirow{2}{*}{Simulation} & 0.1777 & 0.1555 & 0.4111 & 0.4127 & 0.4067 & 0.5921 & 0.2407 & 0.5105 & 239.49 \\ 
		&& (0.0101) & (0.0099) & (0.0184) & (0.0184) & (0.018) & (0.0296) & (0.0123) & (0.0198) &  \\ \cdashline{2-11} 
		& R.E. (\%) & -0.2814 & -0.3215 & -0.1703 & -0.1696 & -0.0984 & -0.1351 & -1.0386 & -0.1175 &  \\ \midrule
		\multirow{4}{*}{\shortstack[l]{KL\\($\alpha\rightarrow 1$)}} & Numerical & 0.1773 & 0.1552 & 0.4108 & 0.4123 & 0.4066 & 0.5920 & 0.2386 & 0.5104 & 13.34 \\ \cdashline{2-11} 
		& \multirow{2}{*}{Simulation} & 0.1780 & 0.1558 & 0.4114 & 0.4129 & 0.4070 & 0.5928 & 0.2407 & 0.5106 & 235.01 \\ 
		&& (0.0101) & (0.0099) & (0.0184) &  (0.0184) & (0.0179) & (0.0296) & (0.0123) & (0.0198) &  \\ \cdashline{2-11} 
		& R.E. (\%) & -0.3933 & -0.3851 & -0.1458 & -0.1453 & -0.0983 & -0.1350 & -0.8725 & -0.0392 &  \\ \midrule
		\multirow{4}{*}{1.001} & Numerical & 0.1774 & 0.1553 & 0.4112 & 0.4127 & 0.4069 & 0.5927 & 0.2387 & 0.5108 & 12.63 \\ \cdashline{2-11} 
		& \multirow{2}{*}{Simulation} & 0.1779 & 0.1557 & 0.4118 & 0.4134 & 0.4073 & 0.5935 & 0.2411 & 0.5113 & 238.03 \\
		&& (0.0101) & (0.0099) & (0.0184) & (0.0184) & (0.0179) & (0.0296) & (0.0123) & (0.0198) &  \\ \cdashline{2-11} 
		& R.E. (\%) & -0.2811 & -0.2569 & -0.1457 & -0.1693 & -0.0982 & -0.1348 & -0.9954 & -0.0978 &  \\ \midrule
		\multirow{4}{*}{1.01} & Numerical & 0.1784 & 0.1562 & 0.4144 & 0.4159 & 0.4100 & 0.5991 & 0.2406 & 0.5145 & 13.39 \\ \cdashline{2-11} 
		& \multirow{2}{*}{Simulation} & 0.1789 & 0.1567 & 0.4150 & 0.4166 & 0.4104 & 0.5999 & 0.2430 & 0.5151 & 235.90 \\
		&& (0.01) & (0.0098) & (0.0184) & (0.0184) & (0.0179) & (0.0298) & (0.0123) & (0.0198) &  \\ \cdashline{2-11} 
		& R.E. (\%) & -0.2795 & -0.3191 & -0.1446 & -0.1680 & -0.0975 & -0.1334 & -0.9877 & -0.1165 &  \\ \midrule 
		\multirow{4}{*}{1.5} & Numerical & 0.2248 & 0.2014 & 0.5807 & 0.5823 & 0.5650 & 0.9971 & 0.3418 & 0.6995 & 13.23 \\ \cdashline{2-11} 
		& \multirow{2}{*}{Simulation} & 0.2253 & 0.2019 & 0.5806 & 0.5828 & 0.5645 & 0.9967 & 0.3411 & 0.7000 & 232.73 \\
		& & (0.0081) & (0.0082) & (0.0178) & (0.0178) & (0.0171) & (0.0437) & (0.0113) & (0.0181) &  \\ \cdashline{2-11} 
		& R.E. (\%) & -0.2219 & -0.2476 & 0.0172 & -0.0858 & 0.0886 & 0.0401 & 0.2052 & -0.0714 &  \\ \midrule
		\multirow{4}{*}{2} & Numerical & 0.2601 & 0.2370 & 0.7330 & 0.7345 & 0.7054 & 1.5699 & 0.4364 & 0.8587 & 12.49 \\ \cdashline{2-11} 
		& \multirow{2}{*}{Simulation} & 0.2606 & 0.2374 & 0.7321 & 0.7348 & 0.7041 & 1.5548 & 0.4335 & 0.8590 & 232.84 \\
		& & (0.0069) & (0.0071) & (0.0181) & (0.0181) & (0.0174) & (0.1445) & (0.0114) & (0.0176) &  \\ \cdashline{2-11} 
		& R.E. (\%) & -0.1919 & -0.1685 &  0.1229 & -0.0408 & 0.1846 & 0.9712 & 0.6690 & -0.0349 &  \\ \bottomrule
	\end{tabular}
\end{table}

From Table~\ref{T1}, these two ways yield similar numerical results, and thus
the two different methods validate each other.
In particular, we feel confident that the sample size and the
replication number for the above examples are large enough in the Monte Carlo
simulation  to estimate the R\'{e}nyi divergence and Kullback--Leibler
divergence.
Note that the R\'{e}nyi divergence increases as $\alpha$ increases, as that
in the i.i.d.\ case. 
Furthermore, the R\'{e}nyi divergence and Kullback--Leibler divergence 
get closer 
when $\alpha \to 1$.
We also observe that our numerical method is stable for the Kullback--Leibler
divergence, whereas it is sensitive to the underlying parameteter change for the
R\'{e}nyi divergence. This may be due to the use of the approximted R\'{e}nyi
divergence.

\section{Conclusion}\label{sec7}
\def\theequation{6.\arabic{equation}}
\setcounter{equation}{0}

In this paper, we study the R\'{e}nyi divergence  for a general HMM, to cover the
Markov switching model and RNN, including the classical HMM as a special
case.
The Kullback--Leibler divergence can be regarded as the limit of $\alpha \to 1$
of the R\'{e}nyi divergence.  Moreover, we  express the Kullback--Leibler
divergence of the general HMM as a top Lyapunov exponent of a well-defined
product of Markovian random matrices. Since the R\'{e}nyi divergence involves
the largest eigenvalue of the associated Markov operator, which is notoriously
difficult to compute, we turn our attention to asymptotic expansions, and
derive an approximated R\'{e}nyi divergence, which can be used for numerical
approximation based on the Fredholm integral equation. 

There are further studies along this line. First, it would be interesting to
approximate the largest eigenvalue $\lambda(\alpha)$ to yield a more accurate
numerical approximation. Second, we will study the case in which  $\{X_n, n
\geq 0\}$ is a general state Markov chain; or a more general HMM to cover
regime switching state space models and regime switching $\mathrm{GARCH}(p,q)$
(stochastic volatility) models. Last, applications of the R\'{e}nyi divergence
and Kullback--Leibler divergence in general HMMs such as model selection,
regularization, and variational inference are also interesting and merit
further investigation.


\bibliographystyle{chicago}
\bibliography{208Aref}~

\section{Appendix}

	\subsection{Theoretical Study of the Invariant Probability in Model~(\ref{example})}
	We consider the  Markov switching model in (\ref{example}),
		$$Y_t=\psi_1 \mu_{X_t}+\psi_2 \mu_{X_{t-1}}+\phi Y_{t-1}+\epsilon_t,$$
		where $\{X_t, t \geq 0\}$ is a Markov chain on a state space ${\cal
		X}=\{0,1\}$,  with transition probability matrix 
		$P_{\bm \theta}=\begin{pmatrix}
		p_{00} & p_{01} \\
		p_{10} & p_{11}
		\end{pmatrix}.$ 
		Here we omit
		${\bm \theta}$ in $p_{ij}$ to simplify the notation. We
		will add ${\bm \theta}_1$ in $p_{ij}$ as $p_{ij}^{{\bm \theta}_1}$ when
		the probability is $P=P_{{\bm \theta}_1}$.
		
		We approach this problem by reformulating model (\ref{example}) as
		a 1-order four-state Markov switch model.
		To do so,  we first create a variable $Z_n$ such that 
		\begin{itemize}
			\item
			$Z_t=0$ if $(X_{t-1},X_{t})=(0,0)$,~~~$Z_t=1$ if $(X_{t-1},X_{t})=(0,1)$,
			\item
			$Z_t=2$ if $(X_{t-1},X_{t})=(1,0)$,~~~$Z_t=3$ if $(X_{t-1},X_{t})=(1,1)$.
		\end{itemize}
		Then $\{Z_t, t \geq 0\}$  is a Markov chain on state space ${\cal
		X}=\{0,1,2,3\}$, with transition probability matrix
		$${P}_{\bm \theta}=
		\begin{pmatrix}
		p_{00} & p_{01} & 0 & 0\\
		0 & 0 & p_{10} & p_{11}\\
			p_{00} & p_{01}& 0 & 0\\
		0 & 0 & p_{10} & p_{11}
		\end{pmatrix}.$$

		Denote the conditional density function of $Y_t$ given $Y_{t-1}$ and
		$(X_{t-1},X_t)=(i,j)$ as $f_{ij,{\bm \theta}}(Y_t|Y_{t-1})$. Then, we can
		represent the joint probability of $Y_1,\dots,Y_n$ as
		\begin{eqnarray}
		p_{\bm \theta}(Y_1,\dots,Y_n)=\|\mathbf{M}_n\dots \mathbf{M}_1\bm{\pi}_{\bm \theta}\|,
		\end{eqnarray}
		where
		$$\bm\pi_{\bm \theta}=
		\begin{pmatrix}
		p_{00} p_{10}/(p_{01}+p_{10})\\
		p_{01} p_{10} /(p_{01}+p_{10})\\
		p_{10} p_{01} /(p_{01}+p_{10})\\
		p_{11} p_{01} /(p_{01}+p_{10})
		\end{pmatrix},
		\quad
		\mathbf M_1=
		\begin{pmatrix}
		f_{00,{\bm \theta}}(Y_1) & 0&0 &0\\
		0& f_{01,{\bm \theta}}(Y_1) & 0& 0\\
		0&0 & f_{10,\bm \theta}(Y_1) &0 \\
		0&0 &0 & f_{11,\bm \theta}(Y_1)
		\end{pmatrix},$$
		and for $t\geq2$,
		$$\mathbf{M}_t=
		\begin{pmatrix}
		p_{00}f_{00,\bm \theta}(Y_t|Y_{t-1}) & 0 & p_{00}f_{00,\bm \theta}(Y_t|Y_{t-1}) & 0\\
		p_{01} f_{01,\bm \theta}(Y_t|Y_{t-1}) &  0 & p_{01} f_{01,\bm \theta}(Y_t|Y_{t-1}) &  0\\
		0 & p_{10}  f_{10,\bm \theta}(Y_t|Y_{t-1}) & 0 & p_{10}  f_{10,\bm \theta}(Y_t|Y_{t-1})\\
		0 & p_{11}f_{11,\bm \theta}(Y_t|Y_{t-1}) & 0 & p_{11} f_{11,\bm \theta}(Y_t|Y_{t-1})\\
		\end{pmatrix}.$$
		
Let
	\begin{align*}
	A_{t,\bm \theta}&=p_{00}(A_{t-1,\bm \theta}+C_{t-1,\bm \theta})f_{00,\bm \theta}(Y_t|Y_{t-1}),
	~B_{t,\bm \theta}=p_{01}(A_{t-1,\bm \theta}+C_{t-1,\bm \theta})f_{01,\bm \theta}(Y_t|Y_{t-1}),\\
	C_{t,\bm \theta}&=p_{10} (B_{t-1,\bm \theta}+D_{t-1,\bm \theta})f_{10,\bm \theta}(Y_t|Y_{t-1}),
	~D_{t,\bm \theta}=p_{11} (B_{t-1,\bm \theta}+D_{t-1,\bm \theta})f_{11,\bm \theta}(Y_t|Y_{t-1}),
	\end{align*}
	with initial value $(A_{1,\bm \theta},B_{1,\bm \theta},C_{1,\bm
	\theta},D_{1,\bm \theta})=\mathbf M_1\bm\pi_{\bm \theta}$. Then we can
	calculate $p_{\bm \theta}(Y_1,\dots,Y_n)$ recursively by
		$$p_{\bm \theta}(Y_1,\dots,Y_n)=A_{n,\bm \theta}+B_{n,\bm \theta}+C_{n,\bm \theta}+D_{n,\bm \theta}.$$

		The log-joint probability of $Y_1,\dots,Y_n$ can be computed as
		$$\log p_{\bm \theta}(Y_1,\dots,Y_n)=\sum_{t=1}^n\log(A_{t, \bm \theta}^*+B_{t,\bm \theta}^*+C_{t,\bm \theta}^*+D_{t,\bm \theta}^*),$$
		where
		\begin{align*}
		A_{t,\bm \theta}^*&=\frac{p_{00}(A_{t-1,\bm \theta}^*+C_{t-1,\bm \theta}^*)}{A_{t-1,\bm \theta}^*+B_{t-1,\bm \theta}^*+C_{t-1,\bm \theta}^*+D_{t-1,\bm \theta}^*}f_{00,\bm \theta}(Y_t|Y_{t-1}),\\
		B_{t,\bm \theta}^*&=\frac{p_{01}(A_{t-1,\bm \theta}^*+C_{t-1,\bm \theta}^*)}{A_{t-1,\bm \theta}^*+B_{t-1,\bm \theta}^*+C_{t-1,\bm \theta}^*+D_{t-1,\bm \theta}^*}f_{01,\bm \theta}(Y_t|Y_{t-1}),\\
		C_{t,\bm \theta}^*&=\frac{p_{10} (B_{t-1,\bm \theta}^*+D_{t-1,\bm \theta}^*)}{A_{t-1,\bm \theta}^*+B_{t-1,\bm \theta}^*+C_{t-1,\bm \theta}^*+D_{t-1,\bm \theta}^*}f_{10,\bm \theta}(Y_t|Y_{t-1}),\\
		D_{t,\bm \theta}^*&=\frac{p_{11} (B_{t-1,\bm \theta}^*+D_{t-1,\bm \theta}^*)}{A_{t-1,\bm \theta}^*+B_{t-1,\bm \theta}^*+C_{t-1,\bm \theta}^*+D_{t-1,\bm \theta}^*}f_{11,\bm \theta}(Y_t|Y_{t-1}),
		\end{align*}
		with initial value $(A_{1,\bm \theta}^*,B_{1,\bm \theta}^*,C_{1,\bm
		\theta}^*,D_{1,\bm \theta}^*)=\mathbf M_1\bm\pi_{\bm \theta}$.

	Then the Kullback--Leibler divergence and R\'{e}nyi divergence can be
	computed as follows. Here we compute only the Kullback--Leibler divergence;
	the computation of the R\'{e}nyi divergence can be done as that in
	(\ref{Gij}). To start with, the Kullback--Leibler divergence can be computed
	by $J_{\bm \theta_1}-J_{\bm \theta}$, where
		$J_{\bm \theta}=\mathbb{E}_{\Pi}\Big[ \log(A_{t,\bm \theta}^*+B_{t,\bm \theta}^*+C_{t,\bm \theta}^*+D_{t,\bm \theta}^*) \Big].$
		
		Define $W_t=(A^*_{t,\bm \theta}+C^*_{t,\bm \theta})/(A^*_{t,\bm
		\theta}+B^*_{t,\bm \theta}+C^*_{t,\bm \theta}+D^*_{t,\bm \theta})$, and
		the stationary density function satisfying
		$$\mathbb{P}(X_{t-1}=j,X_t=k,Y_{t-1}=u,W_t\leq x)=\int_0^x m_{jk}(u,w)dw.$$
		Then we can express $J_{\bm \theta}$ as
		\begin{align*}
		J_{\bm \theta} 
		&=\mathbb{E}_{\Pi}\Big[\log(A_{t,\bm \theta}^*+B_{t,\bm \theta}^*+C_{t,\bm \theta}^*+D_{t,\bm \theta}^*)\Big]\\
		&=\mathbb{E}_{\Pi}\Big[\mathbb{E}_{\Pi}\Big(\log(A_{t,\bm \theta}^*+B_{t,\bm \theta}^*+C_{t,\bm \theta}^*+D_{t,\bm \theta}^*)\Big|X_{t-2},X_{t-1},Y_{t-2},W_{t-1}\Big) \Big]\\
		&=\sum_{i=0}^1\sum_{j=0}^1\sum_{k=0}^1\int_0^1\int_{-\infty}^\infty\mathbb{P}_{\Pi}(X_t=k|X_{t-1}=j)\\
		&\cdot \mathbb{E}_{\Pi}\Big(\log(A_{t,\bm \theta}^*+B_{t,\bm \theta}^*+C_{t,\bm \theta}^*+D_{t,\bm \theta}^*)\Big|X_{t-2}=i,X_{t-1}=j,X_t=k,Y_{t-2}=v,W_{t-1}=w\Big) m_{ij}(v,w)dvdw\\
		&=\int_0^1\int_{-\infty}^\infty\Big[p_{00}^{\bm \theta_1} G_{000}(v,w)+p_{01}^{\bm \theta_1}G_{001}(v,w)\Big]m_{00}(v,w)dvdw \\
		&\qquad+\int_0^1\int_{-\infty}^\infty\Big[p_{10}^{\bm \theta_1}G_{010}(v,w)+p_{11}^{\bm \theta_1}G_{011}(v,w)\Big]m_{01}(v,w)dvdw\\
		&\qquad+\int_0^1\int_{-\infty}^\infty\Big[p_{00}^{\bm \theta_1}G_{100}(v,w)+p_{01}^{\bm \theta_1}G_{101}(v,w)\Big]m_{10}(v,w)dvdw\\
		&\qquad+\int_0^1\int_{-\infty}^\infty\Big[p_{10}^{\bm \theta_1}G_{110}(v,w)+p_{11}^{\bm \theta_1}G_{111}(v,w)\Big]m_{11}(v,w)dvdw,
		\end{align*}
		where
		\begin{align*}
		G_{ijk}(v,w)&=\mathbb{E}_{\Pi}\Big(\log(A_{t,\bm \theta}^*+B_{t,\bm \theta}^*+C_{t,\bm \theta}^*+D_{t,\bm \theta}^*)\Big|X_{t-2}=i,X_{t-1}=j,X_t=k,Y_{t-2}=v,W_{t-1}=w\Big)\\
		&=\int_{-\infty}^\infty\int_{-\infty}^\infty\log\Big[ p_{00}wf_{00,\bm \theta}(y|u)+ p_{01}wf_{01,\bm \theta}(y|u)\\
		&\qquad\qquad\qquad+ p_{10}(1-w)f_{10,\bm \theta}(y|u)+ p_{11}(1-w)f_{11,\bm \theta}(y|u)\Big] f_{ij,\bm \theta_1}(u|v) f_{jk,\bm \theta_1}(y|u)du dy.
		\end{align*}
		
		It remains to approximate $\hat m_{ij}(\cdot)$. Note that
		\begin{align*}
		&\mathbb{P}_{\Pi}(X_{t-1}=j,X_t=k,Y_{t-1}=u,W_t\leq x)\\
		&=\sum_{i=0}^1\sum_{l=0}^1\int_0^1\int_{-\infty}^\infty \mathbb{P}_{\Pi}(X_{t-1}=j,X_t=k,Y_{t-1}=u,W_t\leq x|X_{t-2}=i,X_{t-1}=l,Y_{t-2}=v,W_{t-1}=w)\\
		&\qquad\cdot m_{il}(v,w)dvdw\\
		&=\sum_{i=0}^1\int_0^1\int_{-\infty}^\infty\mathbb{P}_{\Pi}(X_{t-1}=j,X_t=k,Y_{t-1}=u,W_t\leq x|X_{t-2}=i,X_{t-1}=j,Y_{t-2}=v,W_{t-1}=w)\\
		&\qquad\cdot m_{ij}(v,w)dvdw  \\
		&=\sum_{i=0}^1\int_0^1\int_{-\infty}^\infty\mathbb{P}_{\Pi}(X_t=k|X_{t-1}=j) f_{ij,\theta_1}(u|v) \\
		&\qquad\cdot\mathbb{P}_{\Pi}(W_t\leq x|X_{t-2}=i,X_{t-1}=j,X_t=k,Y_{t-1}=u,W_{t-1}=w)m_{ij}(v,w)dvdw\\
		&=\sum_{i=0}^1\int_0^1\int_{-\infty}^\infty \mathbb{P}_{\Pi}(X_t=k|X_{t-1}=j)f_{ij,\theta_1}(u|v) Q_{jk}(x;u,w)m_{ij}(v,w)dvdw,
		\end{align*}
		where
		\begin{align*}
		Q_{jk}(x;u,w)
		&=\mathbb{P}_{\Pi}(W_t\leq x|X_{t-2}=i,X_{t-1}=j,X_t=k,Y_{t-1}=u,W_{t-1}=w)\\
		&=\mathbb{P}_{\Pi}\Bigg(\frac{A_{t,\bm \theta}^*+C_{t,\bm \theta}^*}{A_{t,\bm \theta}^*+B_{t,\bm \theta}^*+C_{t,\bm \theta}^*+D_{t,\bm \theta}^*}\leq x\Bigg|X_{t-2}=i,X_{t-1}=j,X_t=k,Y_{t-1}=u,W_{t-1}=w\Bigg)\\
		&=\mathbb{P}_{\Pi}\Bigg((1-x) p_{00}wf_{00,\bm \theta}(Y_t|u)-x p_{01} w f_{01,\bm \theta}(Y_t|u)+(1-x) p_{10} (1-w)f_{10,\bm \theta}(Y_t|u)\\
		&\qquad\qquad-x p_{11}(1-w)f_{11,\bm \theta}(Y_t|u)\leq0\Bigg|X_{t-1}=j,X_t=k,Y_{t-1}=u\Bigg).
		\end{align*}
		
		It follows that
		\begin{align*}
		m_{00}(u,x)&=\int_0^1\int_{-\infty}^\infty p_{00}^{\bm \theta_1} f_{00,\bm \theta_1}(u|v)\frac{\partial}{\partial x}Q_{00}(x;u,w)m_{00}(v,w)dvdw\\
		&\qquad+\int_0^1\int_{-\infty}^\infty p_{00}^{\bm \theta_1}f_{10,\bm \theta_1}(u|v)\frac{\partial}{\partial x}Q_{00}(x;u,w)m_{10}(v,w)dvdw\\
		m_{01}(u,x)&=\int_0^1\int_{-\infty}^\infty p_{01}^{\bm \theta_1} f_{00,\bm \theta_1}(u|v)\frac{\partial}{\partial x}Q_{01}(x;u,w)m_{00}(v,w)dvdw\\
		&\qquad+\int_0^1\int_{-\infty}^\infty p_{00}^{\bm \theta_1} f_{10,\bm \theta_1}(u|v)\frac{\partial}{\partial x}Q_{01}(x;u,w)m_{10}(v,w)dvdw\\
		m_{10}(u,x)&=\int_0^1\int_{-\infty}^\infty p_{10}^{\bm \theta_1} f_{01,\bm \theta_1}(u|v)\frac{\partial}{\partial x}Q_{10}(x;u,w)m_{01}(v,w)dvdw\\
		&\qquad+\int_0^1\int_{-\infty}^\infty p_{10}^{\bm \theta_1}f_{11,\bm \theta_1}(u|v)\frac{\partial}{\partial x}Q_{10}(x;u,w)m_{11}(v,w)dvdw\\
		m_{11}(u,x)&=\int_0^1\int_{-\infty}^\infty p_{11}^{\bm \theta_1} f_{01,\bm \theta_1}(u|v)\frac{\partial}{\partial x}Q_{11}(x;u,w)m_{01}(v,w)dvdw\\
		&\qquad+\int_0^1\int_{-\infty}^\infty p_{11}^{\bm \theta_1} f_{11,\bm \theta_1}(u|v)\frac{\partial}{\partial x}Q_{11}(x;u,w)m_{11}(v,w)dvdw.
		\end{align*}

		Note that in the model we consider in (\ref{example}), the error term
		$\epsilon_t$ is a normal distribution with probability density function
		\begin{align*}
		f_{ij,\bm \theta}(Y_t|u)&=\frac{1}{\sqrt{\sigma^22\pi}}\exp\Bigg[-\frac{[Y_t-(\psi_2\mu_i+\psi_1\mu_j+\phi u)]^2}{2\sigma^2}\Bigg]\\
		&=\frac{1}{\sqrt{\sigma^22\pi}}e^{-Y_t^2/(2\sigma^2)}\exp\Bigg[\frac{2Y_t(\psi_2\mu_i+\psi_1\mu_j+\phi u)-(\psi_2\mu_i+\psi_1\mu_j+\phi u)^2}{2\sigma^2}\Bigg].
		\end{align*}
		It follows that
		\begin{align*}
		&Q_{jk}(x;u,w)
		=\mathbb{P}_{\Pi}\Bigg((1-x) p_{00}w\exp\Bigg[\frac{2Y_t(\psi_2\mu_0+\psi_1\mu_0+\phi u)-(\psi_2\mu_0+\psi_1\mu_0+\phi u)^2}{2\sigma^2}\Bigg]\\
		&\quad-x p_{01} w \exp\Bigg[\frac{2Y_t(\psi_2\mu_0+\psi_1\mu_1+\phi u)-(\psi_2\mu_0+\psi_1\mu_1+\phi u)^2}{2\sigma^2}\Bigg]\\
		&\quad+(1-x) p_{10} (1-w)\exp\Bigg[\frac{2Y_t(\psi_2\mu_1+\psi_1\mu_0+\phi u)-(\psi_2\mu_1+\psi_1\mu_0+\phi u)^2}{2\sigma^2}\Bigg]\\
		&\quad-x p_{11}(1-w)\exp\Bigg[\frac{2Y_t(\psi_2\mu_1+\psi_1\mu_1+\phi u)-(\psi_2\mu_1+\psi_1\mu_1+\phi u)^2}{2\sigma^2}\Bigg]\\
		&\quad\leq0\Bigg|X_{t-1}=j,X_t=k,Y_{t-1}=u\Bigg).
		\end{align*}

\end{document}